\documentclass[journal]{IEEEtran}
\usepackage[pdftex]{graphicx}
\DeclareGraphicsExtensions{.pdf,.jpeg,.png}
\usepackage{algorithm}
\usepackage{algpseudocode}
\usepackage{footnote}
\usepackage{etoolbox}
\usepackage{cite}
\usepackage{tabu}
\usepackage{setspace}
\usepackage{makecell}
\usepackage{hhline}
\usepackage{epsfig}
\usepackage[usenames, dvipsnames]{color}
\usepackage[roman]{parnotes}

\usepackage{amsthm,amsmath,amssymb}
\usepackage{amsmath}
\usepackage{mathtools}
\usepackage{breqn}
\usepackage{amsfonts}
\usepackage{amssymb}
\usepackage{mathrsfs}
\usepackage{xparse}
\usepackage{multirow}
\usepackage{chngcntr}
\usepackage{dblfloatfix}  
\usepackage{soul} 
\usepackage{caption}
\usepackage{float}
\usepackage{graphicx}

\makeatletter
\def\ps@IEEEtitlepagestyle{%
  \def\@oddhead{\mycopyrightnotice}%
  \def\@oddfoot{\hbox{}\@IEEEheaderstyle\leftmark\hfil\thepage}\relax
  \def\@evenhead{\@IEEEheaderstyle\thepage\hfil\leftmark\hbox{}}\relax
  \def\@evenfoot{}%
}
\def\mycopyrightnotice{%
  \begin{minipage}{\textwidth}
  \scriptsize
  Copyright~\copyright~2021 IEEE. Personal use of this material is permitted. Permission from IEEE must be obtained for all other uses, in any current or future media, including\\reprinting/republishing this material for advertising or promotional purposes, creating new collective works, for resale or redistribution to servers or lists, or reuse of any copyrighted component of this work in other works by sending a request to pubs-permissions@ieee.org. Accepted and published in IEEE Transactions on Green Communications and Networking. Citation information: DOI 10.1109/TGCN.2020.3039282.
  \end{minipage}
}
\makeatother

\makeatletter
\newcommand*{\algrule}[1][\algorithmicindent]{%
  \makebox[#1][l]{%
    \hspace*{.2em}
    \vrule height .75\baselineskip depth .25\baselineskip
  }
}

\newcount\ALG@printindent@tempcnta
\def\ALG@printindent{%
    \ifnum \theALG@nested>0
    \ifx\ALG@text\ALG@x@notext
    \else
    \unskip
    \ALG@printindent@tempcnta=1
    \loop
    \algrule[\csname ALG@ind@\the\ALG@printindent@tempcnta\endcsname]%
    \advance \ALG@printindent@tempcnta 1
    \ifnum \ALG@printindent@tempcnta<\numexpr\theALG@nested+1\relax
    \repeat
    \fi
    \fi
}
\patchcmd{\ALG@doentity}{\noindent\hskip\ALG@tlm}{\ALG@printindent}{}{\errmessage{failed to patch}}
\patchcmd{\ALG@doentity}{\item[]\nointerlineskip}{}{}{} 
\makeatother

\DeclarePairedDelimiter{\ceil}{\lceil}{\rceil}
\DeclareMathOperator{\tr}{tr}
\DeclareMathOperator{\diag}{diag}
\allowdisplaybreaks

\begin{document}
%
\title{An Optimal Low-Complexity Energy-Efficient ADC Bit Allocation for Massive MIMO}

\author{\IEEEauthorblockN{I. Zakir Ahmed$^\star$, Hamid Sadjadpour$^\star$, and Shahram Yousefi$^\ast$\\}
\IEEEauthorblockA{$^\star$ Department of Electrical and Computer Engineering, UC Santa Cruz.\\ $^\ast$Department of Electrical and Computer Engineering, Queen's University, Canada}}


%

\maketitle
\begin{abstract}
Fixed low-resolution \textcolor{black}{Analog to Digital Converters (ADC)} help reduce the power consumption in \textcolor{black}{millimeter-wave Massive Multiple-Input Multiple-Output} (Ma-MIMO) receivers operating at large bandwidths. However, they do not guarantee optimal Energy Efficiency (EE). \textcolor{black}{It has been shown that adopting variable-resolution (VR) ADCs in Ma-MIMO receivers can improve performance with Mean Squared Error (MSE) and throughput while providing better EE. In this paper,} we present an optimal energy-efficient bit allocation (BA) algorithm for Ma-MIMO receivers equipped with VR ADCs under a power constraint. We derive an expression for EE as a function of the Cramer-Rao Lower Bound on the MSE of the received, combined, and quantized signal. An optimal BA condition is derived by maximizing EE under a power constraint. We show that the optimal BA thus obtained is exactly the same as that obtained using the brute-force BA with a significant \textcolor{black}{reduction in computational complexity.} We also study the EE performance and computational complexity of a heuristic algorithm that yields a near-optimal solution.
\end{abstract}


%

%
\IEEEpeerreviewmaketitle

\newcommand{\Xmatrix}{
\begin{bmatrix}
\ddots  & 0     & 0 \\
0  & \frac{1}{{\sigma_i^2}} & 0 \\
0 & 0 & \ddots
\end{bmatrix}}
\newcommand{\Ymatrix}{
\begin{bmatrix}
\ddots  & 0  & 0 \\
0  & \frac{f(b_i)l_i}{\big(1-f(b_i)\big) \sigma_i^2} & 0 \\
0 & 0 & \ddots
\end{bmatrix}}
\newcommand{\Zmatrix}{
\begin{bmatrix}
\ddots  & 0  & 0 \\
0  & \frac{\sigma_i^2}{\sigma_n^2 + \frac{f(b_i)l_i}{\big(1-f(b_i)\big)}} & 0 \\
0 & 0 & \ddots
\end{bmatrix}}
\newcommand{\ZMmatrix}{
\begin{bmatrix}
\ddots  & 0  & 0 \\
0  & \frac{\sigma_i^2}{\sigma_n^2 + \frac{f(b_i)l_i}{\big(1-f(b_i)\big)}} + \frac{1}{p} & 0 \\
0 & 0 & \ddots
\end{bmatrix}}
\newcommand{\Fmatrix}{
\begin{bmatrix}
\ddots  & 0  & 0 \\
0  & 10 & 0 \\
0 & 0 & \ddots 
\end{bmatrix}}
\newcommand{\InvFmatrix}{
\begin{bmatrix}
\ddots  & 0 & 0 \\
0  & f(b_i)\big(1-f(b_i)\big)l_i & 0 \\
0 & 0 & \ddots
\end{bmatrix}}

\section{Introduction}\label{Intro}
\textcolor{black}{Today's telecommunication networks contribute to $2\%$ of the total carbon dioxide emissions \cite{Ganti2, Ganti}. The radio access network contributes about $92\%$ of the total power consumption \cite{SpecArt, BTSPow}. Studies show that 5G base stations require about three times the power of 4G base stations \cite{SpecArt}. One of the 5G standards' goals is to improve the overall network energy efficiency (EE). The 5G standards have set a goal of 100x improvement in network EE compared to the existing 4G-LTE networks \cite{5GReq}. Massive Multiple-Input Multiple-Output (Ma-MIMO) technology is considered both at sub-6Ghz and millimeter wave (mmWave) frequencies. In both scenarios, a large number of antennas help to increase the capacity of the system. Millimeter-wave Ma-MIMO is considered for the back-haul wireless interconnects between the Base Stations (BS), to achieve high throughput and spectral efficiency. However, this comes at the cost of increased power consumption, resulting in poor EE \cite{5GBackHaul,5GBackHaul2}.\\
\indent As envisioned by the 5G standards, network densification ramifications are a complex heterogeneous network (HetNet) consisting of many small- and medium-sized cells, and macrocells. The Single-User (SU) Ma-MIMO framework forms the backbone of communication links between the back-haul HetNet elements \cite{Heter5G}. By splitting the precoding and combining between analog and digital domains (hybrid precoding and combining), the number of RF paths can be reduced considerably as compared to the number of transmit and receive antennas \cite{SigProc,mmPreCom}. Despite adopting hybrid combing at the receiver, the system's overall energy efficiency is poor because the analog to digital converters (ADC) operating at such large bandwidths and high bit-resolution consume a large amount of power \cite{5GBackHaul,SigProc,Rangan}.} In addition to power consumption, high-resolution ADCs operating at high sampling frequencies produce huge amounts of data that are difficult to handle. Using fixed low-resolution ADCs is a popular approach adopted in Ma-MIMO receiver architectures to mitigate large power demands \cite{Jmo}. However, an optimal EE performance is necessary to meet the stringent demands set out by the 5G standards \cite{SpecArt,5GReq}. Adopting variable-resolution (VR) ADCs in Ma-MIMO settings yields such benefits \cite{VarBitAllocJour,Zakir1,Zakir2,Zakir3}.

\subsection{Previous Works} \label{pwork}
Low-resolution ADC MIMO receiver architectures using 1-bit and fixed $n$-bit frameworks have been studied extensively over the last few years \cite{Jmo,Mezghani,Muris,HybArchCap,Uplink}. \textcolor{Black}{Overall, the 1-bit ADC receiver architecture in MIMO receivers has been shown to improve EE; however, at the cost of performance at medium to high SNR regimes for a broad set of system parameters like the number of transmit or receive antennas, the order of modulation used, and channel distribution.} For example, it has been shown that despite improved deployment cost, there is considerable rate loss in the medium to high SNR regimes with 1-bit ADC architectures \cite{Jmo}. It has also been shown that by a small increase in the resolution of ADCs (eg., with 3 bits) on all RF paths, significant performance gains can be achieved for a broad range of system parameters \cite{SvenCap}. \textcolor{black}{Also, there is performance degradation due to channel estimation using low-resolution ADCs \cite{ChanEst}. A practical channel estimation approach under the impact of ADC quantization is considered in \cite{DongR2,DongR3} (in addition to the data transmission stage). The uplink performance evaluation of a multiuser Ma-MIMO system with spatially correlated channels using low-resolution ADCs at the base station is presented in \cite{DongR1}.}\\
\indent All the papers above use fixed-bit-resolution ADCs on the receiver's RF paths. Since the resolutions of ADCs are fixed and low, an optimal EE performance is not guaranteed for a given channel. \textcolor{Black}{From the simulations in \cite{VarBitAllocJour,Zakir1,Zakir2,Zakir3}, it can be seen that by varying the ADC resolutions on each RF path for a given channel \textcolor{black}{condition} and receiver power budget, optimal performance is obtained.} Thus, employing VR ADCs on the receiver's RF paths can be advantageous. \textcolor{black}{The VR ADCs employed should have the ability to change bit resolutions across coherence time. Here, the novel VR ADC architectures and mixed-ADC-bank hardware structures proposed in previous works can be considered \cite{VRadc1,VRadc3}. An ADC Bit Allocation (BA) mechanism that decides on the bit resolution to be used on a given RF path and coherence duration is consequential in achieving optimal EE. Another advantage of employing VR ADCs along with an optimal BA scheme is that a high-resolution ADC can be brought into the signal path during the pilot signal acquisition, thereby removing the ill effects of low-resolution ADCs on channel estimation. Also, the absence of doppler due to the communication between fixed network elements in a wireless backhaul ensures longer coherence durations even at mmWave frequencies \cite{CoherTime,rapaport}. This relaxes the requirement for faster switching of ADC bit resolutions between coherence frames and makes the adoption of VR ADCs in the mmWave wireless backhaul more amicable  \cite{VRadc3}.}\textcolor{black}{ On the other hand, the hardware cost of the novel VR ADC architectures may be higher. However, the energy saving and the long term positive environmental benefits of achieving optimal EE underscores the initial higher cost disadvantage.}\\ 
\indent A BA mechanism based on minimizing the Mean Square Quantization Error (MSQE) under the receiver power constraint is presented in \cite{VarBitAllocJour}. A BA mechanism based on the mean squared error (MSE) minimization under a power constraint using a Genetic Algorithm was proposed in \cite{Zakir1}. An optimal BA based on MSE minimization for a SU mmWave Ma-MIMO channel under a power constraint was derived in \cite{Zakir2}. \textcolor{Black}{A similar Algorithm based on channel capacity maximization was derived in \cite{Zakir3}. In a more recent paper by Kaushik et al., a joint BA and hybrid beamforming strategy is proposed \cite{kaushik2019energy}. In this work, the BA is jointly designed for both digital to analog converts (DAC) and ADCs, along with hybrid precoder and combiner, thus effectively improving the overall EE. It is also shown that the DAC/ADC BA is dynamic during operation and achieves higher EE when compared with existing benchmark techniques that use fixed DAC and ADC bit resolutions \cite{kaushik2019energy}. The authors in \cite{kaushik2019energy} propose a novel alternating direction method of multipliers to optimize hybrid precoder, combiner, and BA matrices jointly for both ADC/DAC, thus achieving lower computational complexity. In the proposed work, we focus mainly on the optimal ADC BA for EE, and hence the computational complexity of our proposed algorithm may not be as good as that of \cite{kaushik2019energy}.}
\subsection{Our Contribution}\label{cwork}
\textcolor{black}{The contributions of this paper are as follows:
\begin {itemize}
\item We propose an ADC BA scheme whose solution is precisely the same as that obtained using the brute-force or exhaustive search (ES) BA with an order of magnitude reduction in multiplication complexity. This provides for optimal EE performance under a power constraint for a SU Ma-MIMO wireless back-haul framework.
\item For the first time, we derive an analytical expression for EE as a function of the Cramer-Rao Lower Bound (CRLB) on MSE of the received, quantized, and combined signal. Using this expression, we derive the proposed ADC BA algorithm.
\item We also propose a heuristic algorithm using simulated annealing (SA) that is near-optimal. The parameters of the SA algorithm can be tuned to trade off the EE optimality and computational complexity.
\end {itemize}}
 
\textit{Notation:}
The column vectors are represented as boldface small letters and matrices as boldface uppercase letters. The primary diagonal of a matrix is denoted as ${\rm diag}(\cdot)$, and all expectations $E[\cdot]$ are over the random variable $\bold{n}$, which is an AWGN vector, i.e., $E[\cdot] = E_{\bold{n}}[\cdot]$. The multivariate normal distribution with mean $\boldsymbol{\mu}$ and covariance $\boldsymbol{\varphi}$ is denoted as $\mathcal{N}(\boldsymbol{\mu},\boldsymbol{\varphi})$ and $\mathcal{CN}(\bold{0},{\boldsymbol{\varphi}})$ denotes a multivariate complex-valued circularly-symmetric Gaussian distribution. The trace of a matrix $\bold{A}$ is shown as $\tr{(\bold{A})}$ and the $N \times N$ identity matrix as $\bold{I}_N$. The term $h(\bold{x})$ defines the differential entropy of a continuous random variable $\bold{x}$. The superscripts $T$ and $H$ denote transpose and Hermitian transpose, respectively. The terms $\mathbb{I}$, $\mathbb{R}$, and $\mathbb{C}$ indicate the set of integer, real, and complex numbers, respectively. 

The rest of this paper is organized as follows. Section \ref{sigmod} describes the system model and parameters. In Section \ref{ba}, we derive the optimal BA conditions for EE. The Section \ref{Sec_Algo} describes the two proposed Algorithms based on the optimal condition derived in Section \ref{ba}. In Section \ref{Sim}, we present the simulation results, and in Section \ref{speed}, we study and compare the computational complexities, followed by the conclusions in Section \ref{conc}. Theorems and their proofs are presented in the Appendices.

\section{Signal Model}\label{sigmod}
The signal model for a typical SU Ma-MIMO transceiver with hybrid precoding and combining is shown in Figure \ref{fig:Fig1new.pdf}. \textcolor{Black}{This signal model forms an underlying framework for wireless backhaul communication link between basestations in a HetNet \cite{5GBackHaul,5GBackHaul2}.} In Figure \ref{fig:Fig1new.pdf}, ${\bold{F}_D}$ and ${\bold{F}_A}$ denote the digital and analog precoders, respectively. Similarly, ${\bold{W}_D^H}$ and ${\bold{W}_A^H}$ represent the digital and analog combiners, respectively. The vector $\bold{x}$ is an $N_s\times1$ transmitted signal vector with unit average power. Let $N_{rt}$ and $N_{rs}$ denote the number of RF chains at the transmitter and  receiver, respectively. Also, $N_t$ and $N_r$ represent the number of transmit and receive antennas, respectively. The channel matrix $\bold{H} = \big[ h_{ij} \big]$ is an \begin{math}(N_r\times N_t)\end{math} matrix representing the line of sight mmWave Ma-MIMO channel with properties defined in \cite{rapaport} (chapter 3, pages 99-125).

\begin{figure}[h]
\centering
\resizebox{.4\textwidth}{!}{
\input{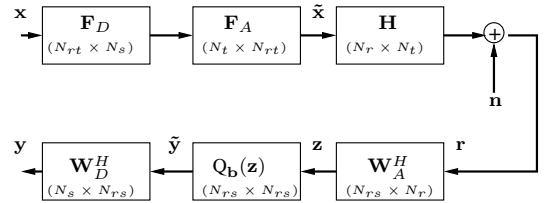}
}
\caption{Signal Model.}
\label{fig:Fig1new.pdf}
\end{figure}

The transmitted signal $\bold{\tilde{x}}$ and the received signal $\bold{r}$ are thus known as ${\bold{\tilde{x}}} = {\bold{F}_A}{\bold{F}_D}{\bold{x}}\text{ and  }{\bold{r}} = {\bold{H}}{\bold{\tilde{x}}}+{\bold{n}}$.
Here, ${\bold{n}}$ is an $N_r\times1$ noise vector of independent and identically distributed (i.i.d.) complex Gaussian random variables such that ${\bold{n}} \sim \mathcal{CN}(\bold{0},{\sigma_n^2}{\bold{I}_{N_r}})$. The received symbol vector $\bold{r}$ is analog-combined with ${\bold{W}_A^H}$ to get ${\bold{z}} = {\bold{W}_A^H}{\bold{r}}$  and later  digitized using a variable-bit quantizer to produce ${\bold{\tilde{y}}} = \text{Q}_{\bold{b}} \big( {\bold{z}} \big) = \bold{W}_{\alpha}\big( {\bold{b}} \big){\bold{z}}+{\bold{n}_q}$ \cite{VarBitAllocJour}. This  signal is combined using the digital combiner ${\bold{W}_D^H}$ to produce the output signal ${\bold{y}} = {\bold{W}_D^H}{\bold{\tilde{y}}}$. The quantizer is modeled as an Additive Quantization Noise Model (AQNM) \cite{Rangan,Rangan2}. 
Here $\bold{b}=[b_1 b_2 b_3 .... b_N]^T$ is a vector whose entries $b_i$ indicate the number of bits (on both I and Q channels) that are allocated to the ADC on RF path $i$. \textcolor{black}{The bits $b_i \in \mathbb{I}$ take values between 1 and $N_b$.} The vector $\bold{n}_q$ has a distribution of $ \mathcal{CN}(\bold{0},{\bold{D}_q^2})$ and is uncorrelated with $\bold{z}$ \cite{Rangan,Rangan2}.

Hence, the relationship between the transmitted signal vector $\bold{x}$ and the received symbol vector $\bold{y}$ at the receiver is given by
\begin{equation}\label{eq5a}
\begin{split}
{\bold{y}} &= {\bold{W}_D^H}{\bold{W}_{\alpha}}{\big( {\bold{b}} \big)}{\bold{W}_A^H}{\bold{H}}{\bold{F}_A}{\bold{F}_D}{\bold{x}} + {\bold{W}_D^H}{\bold{W}_{\alpha}\big( {\bold{b}} \big)}{\bold{W}_A^H}{\bold{n}}\\
&+{\bold{W}_D^H}{\bold{n}_q},
\end{split}
\end{equation}
where the  dimensions of matrices are 
${\bold{F}_D} \in \mathbb{C}^{N_{rt} \times N_s}$, ${\bold{F}_A} \in \mathbb{C}^{N_t \times N_{rt}}$, ${\bold{H}} \in \mathbb{C}^{N_r \times N_t}$, ${\bold{W}_A^H} \in \mathbb{C}^{N_{rs} \times N_r}$, ${\bold{W}_D^H} \in \mathbb{C}^{N_s \times N_{rs}}$, and $\bold{W}_{\alpha}\big( {\bold{b}} \big) \in \mathbb{R}^{N_{rs} \times N_{rs}}$.

With the diagonal BA matrix $\bold{W}_{\alpha}\big( {\bold{b}} \big)$, we intend to design the precoders  ${\bold{F}_D}$, ${\bold{F}_A}$, and Combiners ${\bold{W}_D^H}$, ${\bold{W}_A^H}$, along with the ADC BA ${\bold{W}_{\alpha}\big( {\bold{b}} \big)}$ for a given channel realization $\bold{H}$. We assume perfect CSI at the transmitter. We further assume that 
$N_{rs} = N_s$ and the extension to the case $N_{rs} \ne N_s$ is straightforward.

\section{Energy-Efficient Bit-Allocation Design}\label{ba}
We first present an expression for the CRLB on the MSE that can be achieved on the received, combined, and quantized signal $\bold{y}$ in \eqref{eq5a}. We then derive the expression for the information rate as a function of the CRLB. The CRLB is a function of the hybrid precoder, hybrid combiner, and the BA matrix. We derive the expression for EE using the information rate. An optimal BA condition is arrived by maximizing the EE under a power constraint.
\subsection{CRLB on MSE as a function of BA}\label{ba_mse}
Having designed the precoders such that ${\bold{F}_{\text{opt}}} \approx {\bold{F}_A}{\bold{F_D}}$ with the constraints described in \cite{Zakir2}, we can rewrite ($\ref{eq5a}$) as
\begin{equation}\label{eq9a}
\begin{aligned}
{\bold{y}} &= {\bold{W}_D^H}{\bold{W}_{\alpha}}{\big( {\bold{b}} \big)}{\bold{W}_A^H}{\bold{U}}{\bold{\Sigma}}{\bold{x}} + {\bold{W}_D^H}{\bold{W}_{\alpha}\big( {\bold{b}} \big)}{\bold{W}_A^H}{\bold{n}}+{\bold{W}_D^H}{\bold{n}_q},
\end{aligned}
\end{equation}
with the SVD of the channel matrix as $\bold{H} = \bold{U}\bold{\Sigma}\bold{F}_{\text{opt}}^H$.
Using ($\ref{eq9a}$), we derive the expression for $\mbox{MSE}$ $\delta$ as
\begin{equation}\label{10aa}
\begin{split}
\delta &\triangleq \tr{(E\big[ (\bold{y}-\bold{x})^2\big])}\\
\mbox{MSE}(\bold{x}) &= E\big[ (\bold{y}-\bold{x})^2\big]\\
&= p({\bold{K}}-{\bold{I}_{N_s}})^2 + {\sigma_{n}^2}{\bold{G}\bold{G}^H} + {\bold{W}_D^H}{\bold{D}_q^2}{\bold{W}_D},
\end{split}
\end{equation}
where
${\bold{K}} = {\bold{W}_D^H}{\bold{W}_{\alpha}}{\bold{W}_A^H}{\bold{U}}{\bold{\Sigma}}$, $E[{\bold{x}}{\bold{x}}^H] = p{\bold{I}_{N_s}}$, ${\bold{G}} = {\bold{W}_D^H}{\bold{W}_{\alpha}}{\bold{W}_A^H}$, $E[{\bold{n}}{\bold{n}}^H] = {\sigma_n^2}{\bold{I}_{N_r}}$, $E[{\bold{n}_q}\bold{n}_q^H] = {\bold{D}_q^2}$. Note that $p$ is the average power of symbol $\bold{x}$, ${\bold{D}_q^2} = {\bold{W}_{\alpha}}{\bold{W}_{1-\alpha}}{\text{diag}}[ {\bold{W}_A^H}{\bold{H}}({\bold{W}_A^H}{\bold{H}})^H+{\bold{I}_{N_{rs}}}]$, and $E[{\bold{n}}{\bold{n}_q^H}] = 0$. For simplicity of notation, we refer to $\bold{W}_{\alpha}\big( {\bold{b}} \big)$ as $\bold{W}_{\alpha}$. The expression for the $\mbox{MSE}(\bold{x})$ in ($\ref{10aa}$) can be shown as \cite{Zakir2}
\begin{equation}\label{10ac}
\mbox{MSE}(\bold{x}) = {\sigma_n^2}{\bold{\Sigma}^{-2}} + {\bold{W}_D^H}{\bold{D}_q^2}{\bold{W}_D}.
\end{equation}
The CRLB for \eqref{10ac} is derived as \cite{Zakir2}
\begin{equation}\label{eq26a}
{\bold{I}^{-1}({\bold{\hat{x}}})} = {\sigma_n^2}{\bold{\Sigma}^{-2}}+{\bold{K}^{-1}}{\bold{W}_D^H}{\bold{D}_q^2}{\bold{W}_D}({\bold{K}^H})^{-1}.
\end{equation}
An optimal BA condition based on the CRLB minimization is derived in \cite{Zakir2} by minimizing \eqref{eq26a} with respect to the BA matrix $\bold{W}_{\alpha}$ under a power constraint $P_{\text{ADC}}$.
\begin{equation}\label{bitcond_mse}
\bold{b}^* = \underbrace{\text{argmin}}_{\substack{\bold{b} \in \mathbb{I}^{N_s \times 1}; \\ {P_{\text{TOT}}}\leq{P_{\text{ADC}}}}}\Bigg\{{{\bold{\Sigma}^{-2}}\bigg[ {\sigma_n^2}{\bold{I}_{N_s}} +  {\bold{W}_{\alpha}^{-2}}{\bold{D}_q^2}\bigg]}\Bigg\}.
\end{equation}
 $P_{\text{TOT}}$ is the total power consumed by the ADCs with bit allocation $\bold{b}$ and is shown to equal $2\sum_{i=1}^{N} c{f_s}2^{b_i}$, where $c$ is the power consumed per conversion step and $f_s$ is the sampling rate in Hz \cite{Uplink}. 
\subsection{Energy efficiency as a function of bit allocation}\label{ba_cap}
In this section, we first derive the expression for the information rate of the SU mmWave Ma-MIMO channel encompassing the channel matrix $\bold{H}$, the hybrid precoders ${\bold{F}_D}$, ${\bold{F}_A}$, and the hybrid combiners ${\bold{W}_D^H}$, ${\bold{W}_A^H}$ along with the BA matrix $\bold{W}_{\alpha}$. We then use the information rate to arrive at an expression for EE. Equation \eqref{eq5a} can be simplified as
\begin{equation}\label{eq9a_cap}
\begin{aligned}
{\bold{y}} = {\bold{K}}{\bold{x}} + {\bold{n}_1},
\end{aligned}
\end{equation}
where $\bold{n}_1 = {\bold{W}_D^H}{\bold{W}_{\alpha}}{\bold{W}_A^H}{\bold{n}} + {\bold{W}_D^H}{\bold{n}_q}$. Here $\bold{n}$ is an additive noise vector that is multivariate Gaussian distributed as $\bold{n} \sim \mathcal{CN}(\bold{0},{\sigma_n^2}{\bold{I}_{N_r}})$. \textcolor{Black}{Inspired by \cite{Rangan,Suhas,Rangan2,Uplink}, we assume that $\bold{n}_q$ has Gaussian distribution such that $\bold{n}_q \sim \mathcal{N}(\bold{0},{\bold{D}_q^2})$. This results in $\bold{n}_1$ having the distribution $\mathcal{N}(\bold{0},\bold{\Phi})$ where $\bold{\Phi} = {\sigma_n^2}{\bold{G}}{\bold{G}^H} + {\bold{W}_D^H}{\bold{D}_q^2}{\bold{W}_D}$ \cite{Zakir2}. We assume that  $\bold{x}$ and $\bold{n}_1$ are independent, and is a valid assumption because of the following reasons. The input symbol vector can be modeled as $\bold{x} \sim \mathcal{CN}(\bold{0},p\bold{I}_{N_s})$ \cite{Rangan,VarBitAllocJour}. This can be achieved using efficient Gaussian scramblers \cite{Scramble}. It is straightforward to see that $\bold{n}_1$ and $\bold{x}$ are independent, given that $\bold{n}_1$ and $\bold{x}$ are multivariate Gaussian vectors that are uncorrelated.}
The information rate for the given Ma-MIMO channel can be written as
\begin{equation}\label{eq14a_cap}
\begin{split}
R(\bold{b}) &= I\big(\bold{x}; \bold{y}\big) = h(\bold{y}) - h(\bold{y}|\bold{x})\\
&= h(\bold{y}) - h(\bold{K}\bold{x} + \bold{n}_1|\bold{x}) \overset{(a)} = h(\bold{y}) - h(\bold{n}_1),
\end{split}
\end{equation}
where $I\big(\bold{x}; \bold{y}\big)$ is the mutual information of random variables $\bold{x}$ and $\bold{y}$, and $\bold{K}$ is a function of BA vector $\bold{b}$. (a) holds if and only if both $\bold{n}_q$ and $\bold{x}$ are Gaussian. Hence, ensures $\bold{y}$ is Gaussian.  However, under the assumption that either $\bold{n}_q$ or $\bold{x}$ being non Gaussian, finding a closed form expression of the considered information rate \eqref{eq14a_cap} is an open problem.
Now, if $\bold{y} \in \mathbb{C}^{N_s}$, then the differential entropy $h(\bold{y})$ is less than or equal to $\log_2\det(\pi e \bold{Q})$ with equality if and only if $\bold{y}$ is circularly symmetric complex Gaussian with $E[\bold{y}\bold{y}^H] = \bold{Q}$ \cite{Bengt}. As such,
\begin{equation}\label{eq16a_cap}
\bold{Q} = E \Big[ (\bold{K}\bold{x} + \bold{n}_1)(\bold{K}\bold{x} + \bold{n}_1)^H \Big] = p\bold{K}\bold{K}^H + \bold{\Phi}.
\end{equation}
Note that $\bold{\Phi} = {{\sigma_n^2}{\bold{G}}{\bold{G}^H} + {\bold{W}_D^H}{\bold{D}_q^2}{\bold{W}_D}}$.\\

Thus, the differential entropies $h(\bold{y})$ and $h(\bold{n}_1)$ satisfy
\begin{equation}\label{diffent}
\begin{split}
h(\bold{y}) &\le \log_2\det(\pi e \bold{Q}) = \log_2\det \bigg( \pi e \Big(p \bold{K}\bold{K}^H + \bold{\Phi} \Big) \bigg),\\
h(\bold{n}_1) &\le \log_2\det(\pi e \bold{\Phi}).
\end{split}
\end{equation}
We show that $\bold{n}_1$ is a circularly symmetric jointly Complex Gaussian vector using Theorem \ref{Thm3} in the Appendix. Hence, we can write
\begin{equation}\label{hn_equal}
h(\bold{n}_1) = \log_2\det(\pi e \bold{\Phi}).
\end{equation}
Thus, the information rate $I(\bold{x};\bold{y})$ achieved can be written as
\begin{equation}\label{maxI}
\begin{split}
R(\bold{b}) =  h(\bold{y}) - h(\bold{n}_1) \overset{(b)} = &\log_2\det( \pi e \bold{Q}) - \log_2\det(\pi e \bold{\Phi})\\
= &\log_2\det \Big ( p\bold{K}\bold{K}^H\bold{\Phi}^{-1} + \bold{I}_{N_s} \Big),
\end{split}
\end{equation}
where (b) follows from the assumption that the input symbol vector $\bold{x}$ is circular symmetric Gaussian vector that could be modeled 
as $\bold{x} \sim \mathcal{CN}(\bold{0},p\bold{I}_{N_s})$ \cite{Rangan,VarBitAllocJour}. It is straightforward to see that \eqref{maxI} is a general case of $(17)$ in \cite{SigProc} when the BA is infinite-bits on all ADCs. We simplify ($\ref{maxI}$) to write the information rate as
\begin{equation}\label{maxI_cont}
\begin{split}
R(\bold{b}) &= \log_2\det \Big ( p\bold{K}\bold{K}^H\bold{\Phi}^{-1}\bold{K}\bold{K}^{-1} + \bold{K}\bold{K}^{-1} \Big)\\
&= \log_2 p^{N_s}\det \Big ( \bold{K}^H\bold{\Phi}^{-1}\bold{K} + \frac{1}{p}\bold{I}_{N_s} \Big)\\
&= {N_s}\log_2 p + \log_2\det \Big ( ({\bold{I}^{-1}({\bold{\hat{x}}})})^{-1} + \frac{1}{p}\bold{I}_{N_s} \Big).
\end{split}
\end{equation}
Note that ${\bold{I}^{-1}({\bold{\hat{x}}})}$ is the CRLB (15) in \cite{Zakir2} achieved by the MSE $\delta$ in \eqref{10aa}.
Now, we define EE as a function of BA as \cite{EEdef}
\begin{equation}\label{ee_eq1}
\begin{split}
\eta_{EE}(\bold{b}) &= \frac{R(\bold{b})}{p(\bold{b})}\text{ (bits/Hz/Joule)}\\
&= \frac{{N_s}\log_2 p + \log_2\det \Big ( ({\bold{I}^{-1}({\bold{\hat{x}}})})^{-1} + \frac{1}{p}\bold{I}_{N_s} \Big)}{P_T + P_R + 2 \sum_{i=1}^{N} c{f_s}2^{b_i}}, 
\end{split}
\end{equation}
where $p(\bold{b})$ is the total power consumed. Here $P_T$, $P_R$ are the power consumed at the transmitter and receiver respectively. The net ADC power consumption is $\big(\ 2\sum_{i=1}^{N} c{f_s}2^{b_i}\big)\ $. The expression for $p(\bold{b})$ can be effectively written as
\begin{equation}\label{ee_eq2}
\begin{split}
p(\bold{b}) = 2cf_s \times (\frac{P_T + P_R}{2cf_s} + \sum_{i=1}^{N} 2^{b_i}).
\end{split}
\end{equation}
The transmitter power can be modeled as $P_T = \frac{P_{\text{out}}}{\eta_{PA}}+P_{\text{CIR}}$ \cite{Ganti,Ganti2}. The terms $P_{\text{out}}$, $\eta_{PA}$, and $P_{\text{CIR}}$ represent the transmit power, efficiency of the power amplifier, and basestation circuit power respectively. The receiver power is modeled as $P_R = N_rN_sP_{\text{PS}} + N_rP_{\text{LNA}} + N_sP_{\text{VCO}}$. The terms $P_{\text{PS}}$, $P_{\text{LNA}}$, and $P_{\text{VCO}}$ correspond to the power consumed by a single device phase shifter, Low Noise Amplifier and local oscillator respectively \cite{SigProc}.\\
\indent It is to be noted that the power consumption attributed towards the BA algorithm itself is highly hardware and implementation dependent. To this effect, we consider the computational analysis of the proposed algorithm in terms of number of multiplications and additions, which is discussed in Section \ref{speed}.
\subsection{Hybrid combiner structure}\label{comb_dsgn}
Phase shifters or splitters impose constraints on the design of the analog combiner $\bold{W}_A^H$ \cite{SigProc}. We express the constrained analog combiner as $\bold{\tilde{W}}_A^H$. The digital combiner compensates the imperfections in the analog combiner, that is ${\bold{W}_A^H} = {\bold{W}_D}{\bold{\tilde{W}}_A^H}$ (20) in \cite{Zakir2}. We design the constrained analog combiner ${\bold{\tilde{W}}_A^H}$ and the digital combiner ${\bold{W}_D}$, such that ${\bold{W}_A^H} = {\bold{U}^H} = {\bold{W}_D}{\bold{\tilde{W}}_A^H}$. This is obtained by solving the optimization problem using method described in \cite{PreDsgn}. 
\begin{equation}\label{comb_desgn}
\begin{aligned}
({\bold{\tilde{W}}_A^{opt}},{\bold{W}_D^{opt}}) = & \underbrace{\text{argmin}}_{{\bold{\tilde{W}}_A},{\bold{W}_D}}{\lVert {\bold{U} - {{\bold{\tilde{W}_A}{\bold{W}_D^H}}} \rVert }}_F,\\
\text{ such that } & {\bold{\tilde{W}_A}}\in{\mathcal{W}_{RF}}, {\lVert {{\bold{W}_D^H}{\bold{\tilde{W}}_A}} \rVert }_F^2 = N_s
\end{aligned}
\end{equation}
$\mathcal{W}_{RF}$ is the set of all possible analog combiners architecture based on phase shifters. This includes all possible $N_r \times N_s$ matrices with constant magnitude entries.
\subsection{Maximizing the EE}\label{cap_max}
Let $\bold{b}^*$ be the optimal BA that maximizes the EE in \eqref{ee_eq1}, where
\begin{equation}\label{maxcap}
\begin{split}
&\eta_{EE}(\bold{b}) =\\
&\underbrace{\text{max}}_{\substack{\bold{b}^*,{P_{\text{TOT}}}\leq{P_{\text{ADC}}}}} \Bigg\{ \frac{{N_s}\log_2p + \log_2  \det \Big ( ({\bold{I}^{-1}({\bold{\hat{x}}})})^{-1} + \frac{1}{p}\bold{I}_{N_s} \Big)}{p(\bold{b})} \Bigg\}.
\end{split}
\end{equation}
Thus $\bold{b}^*$ is derived as
\begin{equation}\label{bitcond}
\bold{b}^* = \underbrace{\text{argmax}}_{\substack{\bold{b} \in \mathbb{I}^{N_s \times 1}, \\ {P_{\text{TOT}}}\leq{P_{\text{ADC}}}}}\Bigg\{ \frac{1}{p(\bold{b})} \log_2  \det \Big ( ({\bold{I}^{-1}({\bold{\hat{x}}})})^{-1} + \frac{1}{p}\bold{I}_{N_s} \Big) \Bigg\}.
\end{equation}
By substituting $\bold{K}$ into ($\ref{eq26a}$) and by designing the structure of the hybrid combiner as described earlier, we can simplify the expression for CRLB as
\begin{equation}\label{crlb_cont}
\begin{split}
{\bold{I}^{-1}({\bold{\hat{x}}})} &= {\sigma_n^2}{\bold{\Sigma}^{-2}}+\\
&{\bold{\Sigma}}^{-1}{\bold{U}^H}({\bold{W}_A^H})^{-1}{\bold{W}_{\alpha}^{-1}}{\bold{D}_q^2}{\bold{W}_{\alpha}^{-1}}{\bold{W}_A^{-1}}{\bold{U}}{\bold{\Sigma}}^{-1}\\
&= {\sigma_n^2}{\bold{\Sigma}^{-2}}+{\bold{\Sigma}}^{-2}{\bold{W}_{\alpha}^{-2}}{\bold{D}_q^2}.
\end{split}
\end{equation}
We now compute the Inverse of CRLB $\Big({\bold{I}^{-1}({\bold{\hat{x}}})}\Big)^{-1}$ as 
\begin{equation}\label{inv_crlb}
\begin{split}
&\Big({\bold{I}^{-1}({\bold{\hat{x}}})}\Big)^{-1} = \Big({\sigma_n^2}{\bold{\Sigma}^{-2}}+{\bold{\Sigma}}^{-2}{\bold{W}_{\alpha}^{-2}}{\bold{D}_q^2}\Big)^{-1}\\
&\text{              }= {\rm diag}\bigg( \frac{\sigma_1^2}{\sigma_n^2 + g(b_1)l_1}, \cdots, \frac{\sigma_{N_s}^2}{\sigma_n^2 + g(b_{N_s})l_{N_s}}\bigg),
\end{split}
\end{equation}
Substituting $\Big({\bold{I}^{-1}({\bold{\hat{x}}})}\Big)^{-1}$ in ($\ref{bitcond}$), we have
\begin{equation}\label{bitcond_cont}
\begin{split}
\bold{b}^* = \underbrace{\text{argmax}}_{\substack{\bold{b} \in \mathbb{I}^{N_s \times 1}, \\ {P_{\text{TOT}}}\leq{P_{\text{ADC}}}}} \frac{1}{p(\bold{b})} \sum_{i=1}^{N_s} \bigg\{ \log_2  \Big( q(b_i) + 1 \Big) \bigg\},
\end{split}
\end{equation}
where $q(b_i) = \frac{p\sigma_i^2}{\sigma_n^2 + g(b_i)l_i}$.
The term $\log_2  \Big( q(b_i) + 1 \Big)$ can be expanded for two scenarios given below.\\
\textit{Case 1:}
For the case of $0 \leq q(b_i) < 1$,  we have $\log_2  \Big( q(b_i) + 1 \Big) \simeq \frac{q(b_i)}{\ln2}$. \textcolor{Black}{For proof refer to Lemma $\ref{lemm1}$ in the Appendix.}
Thus, the maximization in $\eqref{bitcond_cont}$ can be written as
\begin{equation}\label{bitcond_cont2}
\bold{b}^* = \underbrace{\text{argmax}}_{\substack{\bold{b} \in \mathbb{I}^{N_s \times 1}, \\ {P_{\text{TOT}}}\leq{P_{\text{ADC}}}}} \frac{1}{p(\bold{b})} \sum_{i=1}^{N_s} \frac{p\sigma_i^2}{\sigma_n^2 + g(b_i)l_i}.
\end{equation}\\
\textit{Case 2:} 
For the case $1 \leq q(b_i) < \infty$, we show that  $\log_2  \Big( q(b_i) + 1 \Big) = \Bigg(1-\frac{1}{q(b_i)}\Bigg)P + L(p,\sigma_i^2, \sigma_n^2)$. For proof refer to Lemma $\ref{lemm2}$ in the Appendix. $P$ and $L(p,\sigma_i^2, \sigma_n^2)$ are independent of $b_i$. Hence, the maximization in $\eqref{bitcond_cont}$ can be simplified to
\begin{equation}\label{bitcond_cont1}
\bold{b}^* = \underbrace{\text{argmax}}_{\substack{\bold{b} \in \mathbb{I}^{N_s \times 1}, \\ {P_{\text{TOT}}}\leq{P_{\text{ADC}}}}} \frac{1}{p(\bold{b})} \sum_{i=1}^{N_s} \Bigg(1-\frac{1}{q(b_i)}\Bigg)\\ 
\end{equation}
Combining the two scenarios, the $\bold{b}^*$ that guarantees optimal EE performance under a power constraint $p(\bold{b}^*) \le  P_{\text{ADC}}$ can be written as
\begin{equation}\label{optee_sol}
\begin{split}
\bold{b}^* = \underbrace{\text{argmax}}_{\substack{\bold{b} \in \mathbb{I}^{N_s \times 1}, \\ {P_{\text{TOT}}}\leq{P_{\text{ADC}}}}} \frac{1}{p(\bold{b})} \Bigg\{ \sum_{b_i \in \mathcal{X}} q(b_i) + \sum_{b_i \in \mathcal{Y}} \bigg(1 - \frac{1}{q(b_i)} \bigg) \Bigg\},
\end{split}
\end{equation}
where $\mathcal{X} = \Big\{ b_i \mid q(b_i) < 1\Big\}$, $\mathcal{Y} = \Big\{ b_i \mid q(b_i)  \ge 1\Big\}$, and $| \mathcal{X} | + | \mathcal{Y} | = N_s$.
\section{Bit Allocation Algorithm}\label{Sec_Algo}
We propose two algorithms to solve the optimal EE condition derived in \eqref{optee_sol}: \textit{(i)} An algorithm that ensures optimal BA \textit{(ii)} A simulated annealing based heuristic technique yielding  near-optimal solution. We described the algorithms below.

\subsection{Algorithm for optimal solution (Q-search)}
The term $q(b_i) = \frac{p\sigma_i^2}{\sigma_n^2 + g(b_i)l_i}$ is evaluated and stored. Here, ${\sigma_i}$ is the diagonal element of ${\bold{\Sigma}}$, ${\sigma_n^2}$ is the noise power, $g(b_i)=\frac{f(b_i)}{1-f(b_i)}$ where $f(b_i)$ is depends on the quantization error on the $i^{th}$ RF path \cite{VarBitAllocJour}. The values for $f(b_i)$ are indicated in \cite{Uplink} and $l_i$ is the $i^{th}$ element of ${\rm diag}(\bold{I}_{N_s}+\bold{W}_D^H\bold{\Sigma}^2\bold{W}_D)$.
For a given $N_s$ and $N_b$, we form a set $B_{\text{set}}$ of all possible $\bold{b}_j$'s that satisfy the ADC power budget $P_{\text{ADC}}$.
\begin{equation}\label{eq35a}
\begin{split}
B_{\text{set}} \triangleq \Big\{ &\bold{b}_j = {\big[ b_{j1}, b_{j2}, \dots, b_{jN_s}  \big]}^T \text{ for } 0 \leq j < N_b^{N_s} \mid \\
&1 \le b_{ji} \le N_b \text{ and } \sum_{i=1}^{N_s} cf_s2^{b_{ji}} \leq P_{\text{ADC}} \Big\}.
\end{split}
\end{equation}
We call this the Q-search method as described in Algorithm~$\ref{AlgoQS}$.
\begin{algorithm}[t]
  \caption{Q-search Algorithm}\label{AlgoQS}
  \begin{algorithmic}[1]
   \small
    \Procedure{Q-search}{$B_{\text{set}}$,$N_s$,$\text{Q}(N_b,N_s)$,$\text{Ptot}(\text{sizeof}(B_{\text{set}}))$}
      \State $B_{\text{set}} \gets \text{Solution Space}$
      \State $N_s\gets \text{Number of RF paths}$
      \State $\text{Q}(N_b,N_s) \gets \text{Table precomputed using \eqref{optee_sol}.}$
      \State $\text{Ptot}(\text{sizeof}(B_{\text{set}})) \gets \text{Table of } -\log_2 \big(p(\bold{b}_j)\big) \forall \bold{b}_j \in B_{\text{set}}.$
      \For{\texttt{j=0;j++;}{until \texttt{j<}sizeof($B_{\text{set}}$)}}
          \State $m \gets \sum_{i=1}^{N_s} \text{Q}(\bold{b}_j(i),i)$
           \State $ \footnotemark p \gets  \text{Ptot}(\bold{b}_j)$
          \State $\footnotemark K_f(b_j) \gets \Call{ShiftLeft}{1, (\log_2(m) + p)}$
     \EndFor
     \State $index\gets \text{max}(K_f)$
     \State $\bold{b^*}\gets B_{\text{set}}\text{ at } index$
     \State \textbf{return} $\bold{b^*}$ \Comment{Optimal Bit Allocation Vector}
    \EndProcedure
  \end{algorithmic}
  \begin{algorithmic}[2]
  \small
  \Procedure{ComputeQ}{$p$,$\sigma_n^2$,$S(N_s)$,$g(N_b)$,$l(N_s)$,$N_b$,$N_s$}
  \State $p \gets \text{Received Signal Power}$
  \State $\sigma_n^2 \gets \text{Noise Power}$
  \State $S(N_s) \gets \text{Table of the square of the singular Values of }\bold{H}.$
  \State $g(N_b) \gets \text{Table of quantization Errors}$ \Comment{Refer \cite{Zakir2,Uplink}}
  \State $l(N_s) \gets \text{Table containing $\diag(\bold{I}_{N_s}+\bold{W}_D^H\bold{\Sigma}^2\bold{W}_D)$}$
  \State $N_b \gets \text{ADC bit range}$
  \State $N_s \gets \text{Number of RF paths}$
  \For{\texttt{i=1;i++;}{until \texttt{i}$\le N_s$}}
  	\For{\texttt{$b_i$=1;$b_i$++;}{until $b_i \le N_b$}}
		\State $q \gets \frac{p\sigma_i^2}{\sigma_n^2 + g(b_i)l(i)}$
		\State $\text{Q}(b_i,i) \gets q \text{ if }q < 1$
		\State $\text{Q}(b_i,i) \gets \bigg( 1 - \frac{1}{q}\bigg) \text{ if }q \ge 1$
  	\EndFor
  \EndFor
  \State \textbf{return} $\text{Q}(N_b,N_s)$
  \EndProcedure
  \end{algorithmic}
\end{algorithm}
\footnotetext[1]{$p = -\log_2(2cf_s) -\log_2\big( \frac{P_T + P_R}{2cf_s} + \sum_{i=1}^{N_s} 2^{b_i}\big).$}
\footnotetext[2]{$\log_2()$ is indexed using lookup table \cite{NumC}}

\vspace{-5mm}
\subsection{Simulated annealing}
The SA is a metaheuristic technique used to solve global optimization problems. While it does not guarantee an optimal solution, tuning its parameters such as the cooling factor can ensure near-optimal solutions \cite{SimAn2}. The SA algorithm has a reduced complexity compared to the Q-search method and is discussed in Section \ref{speed}. The details of the Algorithm~$\ref{AlgoCRLB}$  presented below can be found in \cite{SimAn2}.\\
\begin{algorithm}[!htb]
\caption{Simulated Annealing}\label{AlgoCRLB}
  \begin{algorithmic}[1]
    \small
    \Procedure{SA}{$B_{\text{set}}$,$N_s$,$\text{Q}(N_b,N_s)$,$\text{P}(\text{sizeof}(B_{\text{set}}))$,$T_0$,$r$,$m$}
      \State $N_s\gets \text{Number of spatial-multiplexed paths}$
      \State $B_{\text{set}} \gets \text{Solution Space}$
      \State $\text{Q}(N_b,N_s) \gets \text{Table precomputed using \eqref{optee_sol}.}$
      \State $\text{P}(\text{sizeof}(B_{\text{set}})) \gets \text{Precomputed total power } \forall \bold{b}_j \in B_{\text{set}}.$
      \State $T_0 \gets \text{Initial Temperature}$
      \State $r \gets \text{Cooling factor}$
      \State $m \gets \text{Number of searches at a given temperature }t$
      \State $t \gets T_0$ \text{ Initialize Temperature}
      \State{$\bold{b}_{test} \gets  \text{Select a initial solution from }B_{\text{set}}$}
      \State{$cost \gets \frac{1}{\text{P}(\bold{b}_{test})}\sum_{i=1}^{N_s} \text{Q}(\bold{b}_{test}(i),i)$}
      \State{$(c_{opt}, \bold{b^*}) \gets (cost, \bold{b}_{test})$}
      \While {$t > 1.0$}
		\For{$m\text{ times}$}
			\State{$\bold{b}_{new} \gets  SearchNeighbour(\bold{b}_{test},B_{\text{set}})$}
			\State{$c_{new} \gets  \frac{1}{\text{P}(\bold{b}_{test})}\sum_{i=1}^{N_s} \text{Q}(\bold{b}_{new}(i),i)$}
			\State{$\delta \gets  c_{new}-cost$}
			\State{$P_a \gets \frac{1}{1+ e^{-\frac{\delta}{t}}}$}
			\If {$rand() \le P_a$} \Comment{$rand() \sim \mathcal{U}(0,1)$}
				\State{$(cost, \bold{b}_{test}) \gets (c_{new}, \bold{b}_{new})$}
				\If {$c_{new} > c_{opt}$}
					\State{$(c_{opt}, \bold{b^*}) \gets (c_{new}, \bold{b}_{new})$}			
				\EndIf
			\EndIf
		\EndFor
	\State{$t \gets rT$}
	\EndWhile
	\State \textbf{return} $\bold{b^*}$ \Comment{Optimal bit allocation vector}
    \EndProcedure
    \end{algorithmic}
    \begin{algorithmic}[2]
    \small
    \Procedure{SearchNeighbour}{$\bold{b}_{test}$,$B_{\text{set}}$}
    \State $\bold{b}_{test} \gets \text{ Current solution}$
    \State $B_{\text{set}} \gets \text{ Solution space}$ 
      \State $\bold{b}_{new} \gets \text{LookupNewSolution}(randn(),\bold{b}_{test})$ 
      \State \textbf{return} $\bold{b}_{new}$ \Comment{ Return new solution}
    \EndProcedure
  \end{algorithmic}
\end{algorithm}
\vspace{-5mm}
\section{Simulations and Results} \label{Sim}
We simulate the mmWave channel using the NYUSIM channel simulator for two channel scenarios. In one, we have 2 dominant scatters, and in other we have one dominant scatter \cite{nyusim}. The parameter configurations for the simulations is given in Table $\ref{nyusimtab}$. We consider \textcolor{black}{$N_b = 4$, }$N_s = 8$, and $N_s = 12$ in our simulations. 
We run the simulations to evaluate the EE ($\eta_{EE}$) derived in \eqref{ee_eq1} (Figures \ref{fig:crlb_EE_Nr8_H64by32.eps}-\ref{fig:crlb_EE_Nr12_H64by32_S2.eps}), and the information rate $R$ derived in \eqref{maxI_cont} (Figures \ref{fig:crlb_R_Nr8_H64by32.eps}-\ref{fig:crlb_R_Nr12_H64by32_S2.eps}) at various SNRs for $N_s = 8$ and $N_s = 12$. Monte-Carlo simulations are run with 1-bit ADCs (represented using lines-(a)) and 2-Bit ADCs (line-(b)) across all RF paths. The simulations are also run using the proposed Q-Search (line-(d)), SA (lines-(e) and (f)), and ES (line-(c)) method.\\
\indent The Q-search Algorithm always yields the optimal BA. That is, the BA solution evaluated using the proposed Q-search method is exactly the same as that of the ES method. The performance of the SA Algorithm with cooling factors 0.9 and 0.5 are indicated using the lines (e) and (f), respectively. We observe that the BA solution evaluated using SA is near-optimal with significantly reduced computational complexity compared to the Q-search method. The computational complexity analysis for these methods are discussed in Section \ref{speed} and summarized in Table \ref{tab:CRLBTab1}.
\begin{figure*}[!htb]
\centering
\begin{minipage}[b]{0.45\linewidth}
\centering
\includegraphics[width=0.8\textwidth]{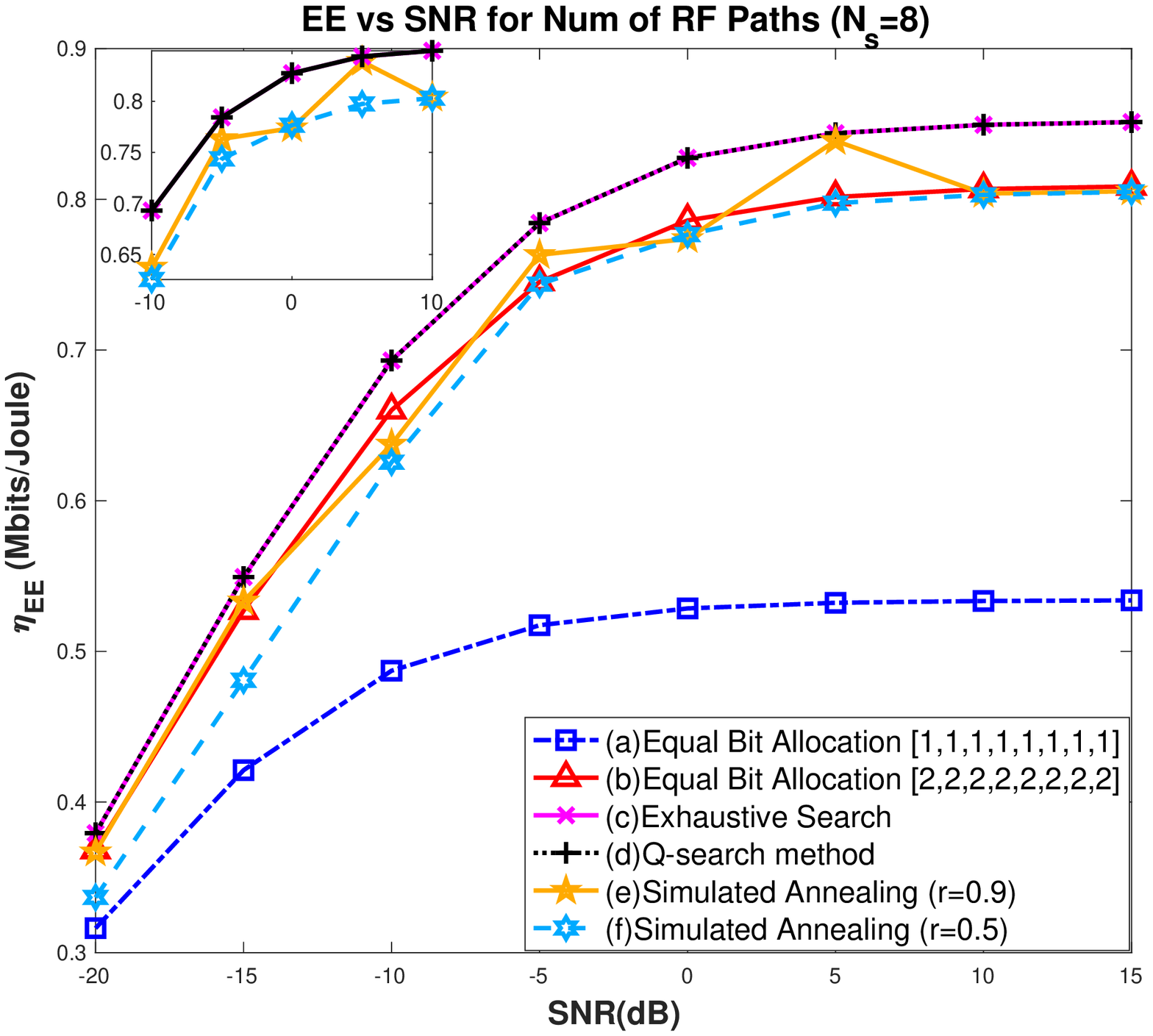}
\captionsetup{justification=centering, font=scriptsize, labelfont={color=Black}}
\caption{\textcolor{Black}{Energy efficiency vs. SNR for $N_s=8$ with 2 dominant scatterers.}}
\label{fig:crlb_EE_Nr8_H64by32.eps}
\end{minipage}
\quad 
\begin{minipage}[b]{0.45\linewidth}
\centering
\includegraphics[width=0.85\textwidth]{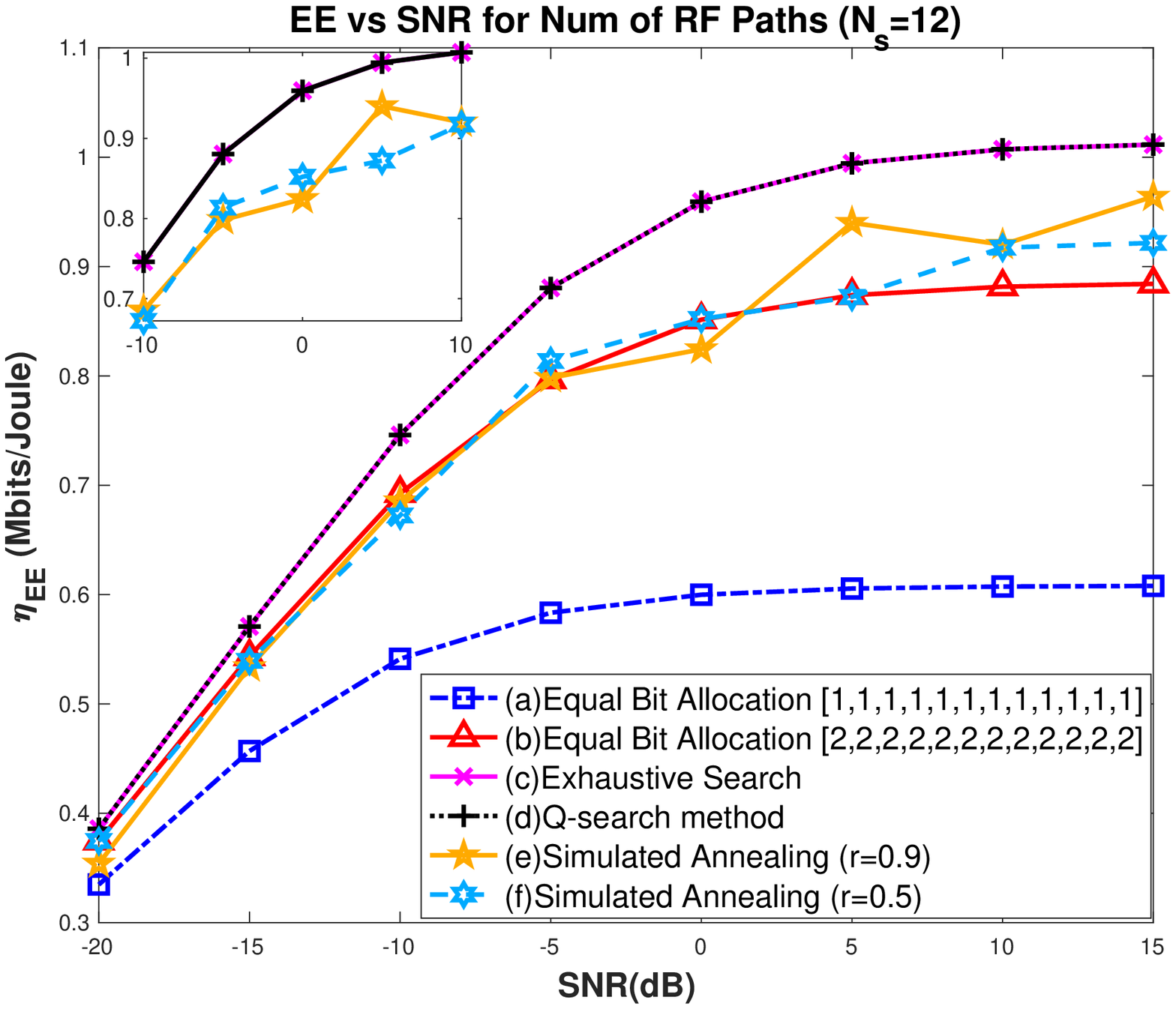}
\captionsetup{justification=centering, font=scriptsize, labelfont={color=Black}}
\caption{\textcolor{Black}{Energy efficiency vs. SNR for $N_s=12$ with 2 dominant scatterers.}}
\label{fig:crlb_EE_Nr12_H64by32.eps}  
\end{minipage}
\begin{minipage}[b]{0.45\linewidth}
\centering
\includegraphics[width=0.8\textwidth]{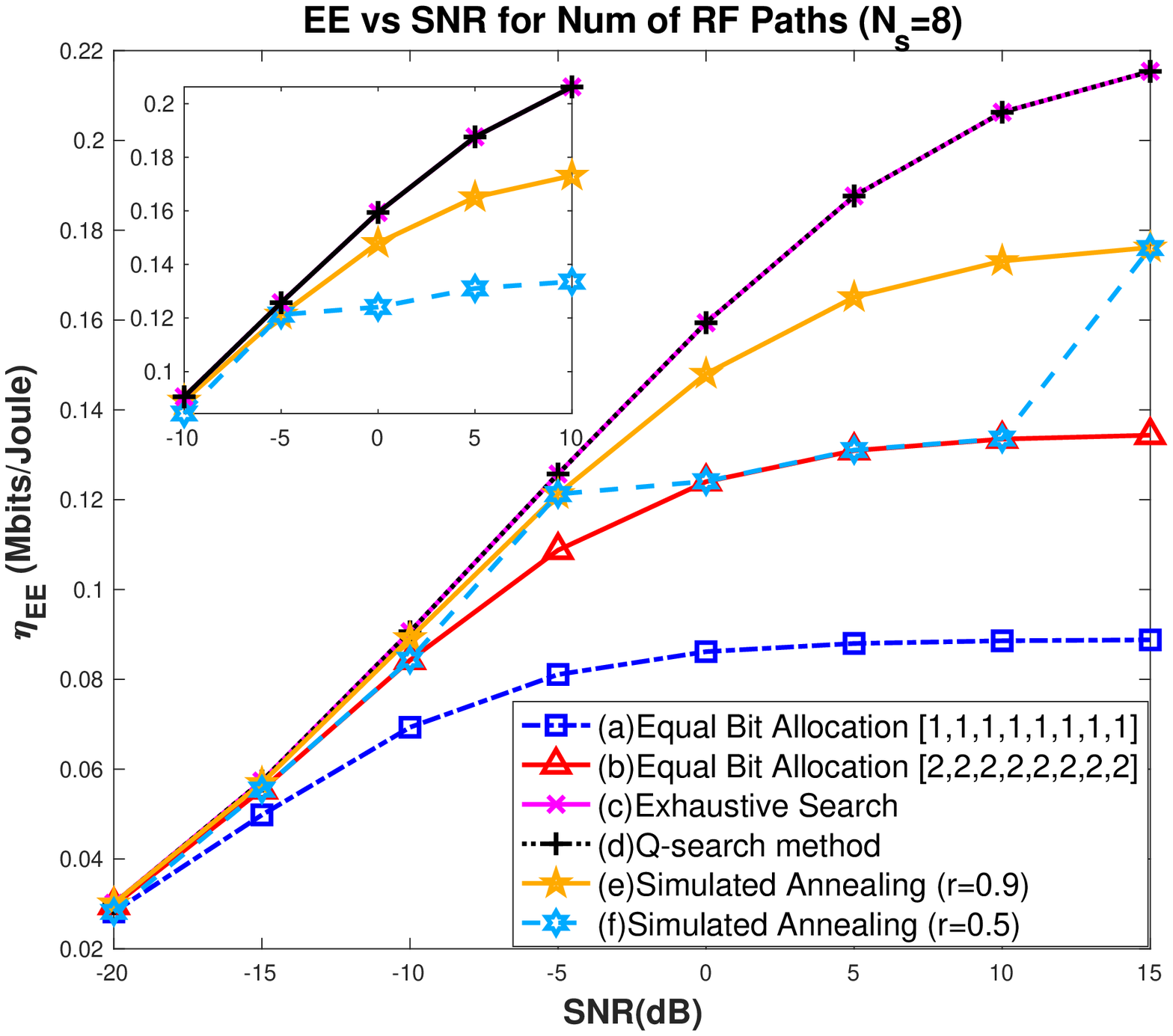}
\captionsetup{justification=centering, font=scriptsize, labelfont={color=Black}}
\caption{\textcolor{Black}{Energy efficiency vs. SNR for $N_s=8$ with a single dominant scatterer.}}
\label{fig:crlb_EE_Nr8_H64by32_S2.eps}
\end{minipage}
\quad 
\begin{minipage}[b]{0.45\linewidth}
\centering
\includegraphics[width=0.85\textwidth]{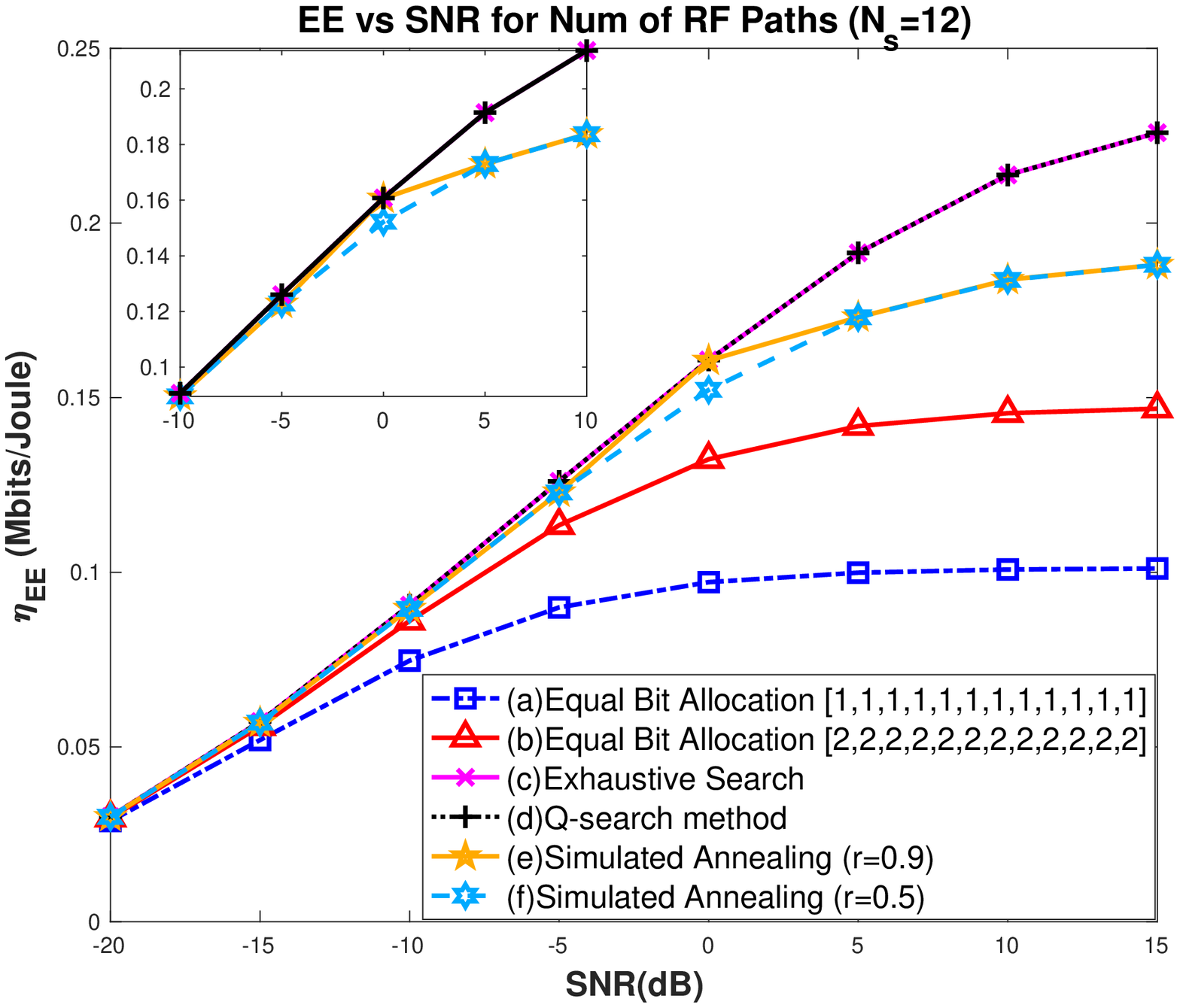}
\captionsetup{justification=centering, font=scriptsize, labelfont={color=Black}}
\caption{\textcolor{Black}{Energy efficiency vs. SNR for $N_s=12$ with a single dominant scatterer.}}
\label{fig:crlb_EE_Nr12_H64by32_S2.eps}  
\end{minipage}
\begin{minipage}[b]{0.45\linewidth}
\centering
\includegraphics[width=0.8\textwidth]{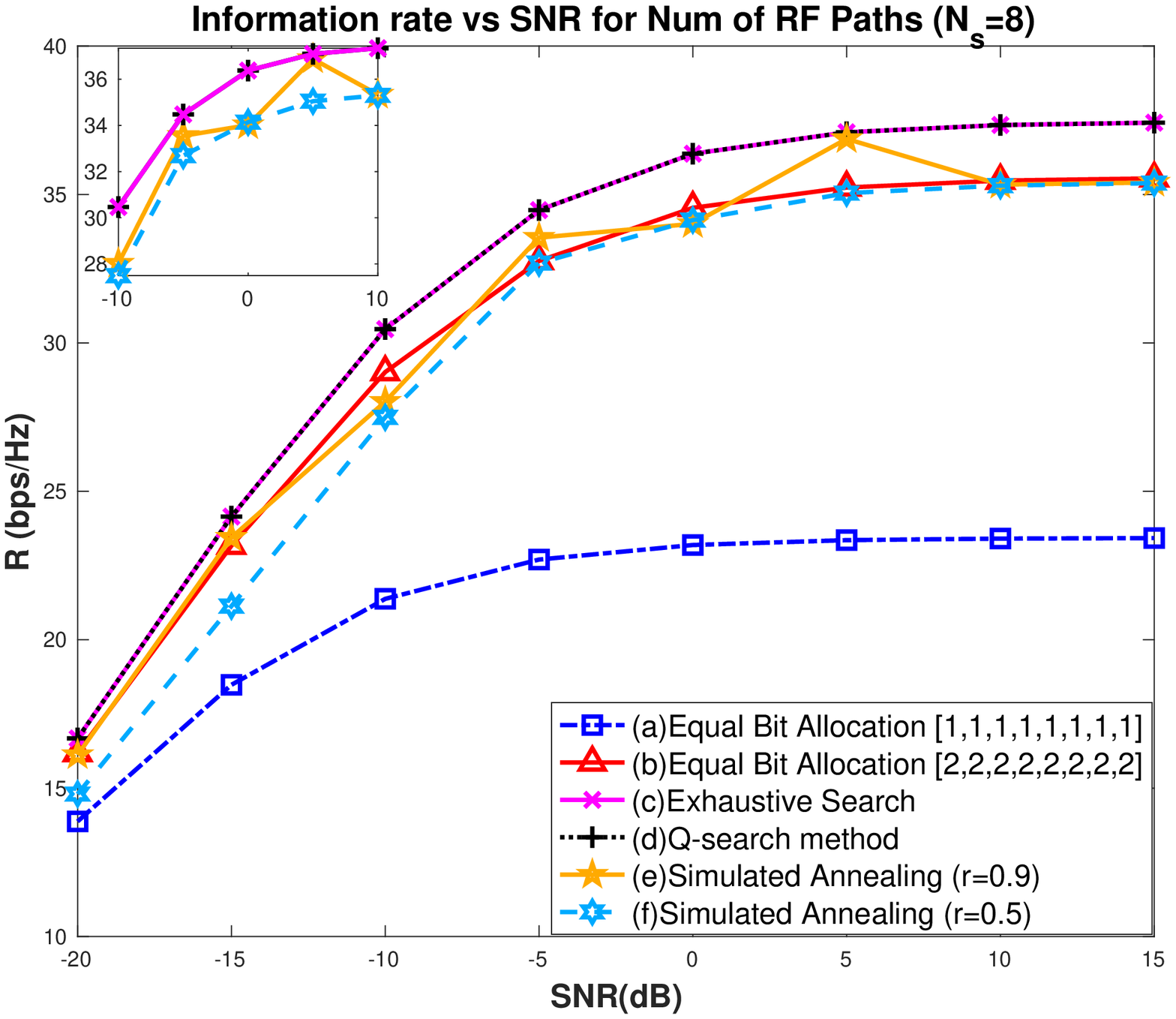}
\captionsetup{justification=centering, font=scriptsize, labelfont={color=Black}}
\caption{\textcolor{Black}{Information rate vs. SNR for $N_s=8$ with 2 dominant scatterers.}}
\label{fig:crlb_R_Nr8_H64by32.eps}
\end{minipage}
\quad 
\begin{minipage}[b]{0.45\linewidth}
\centering
\includegraphics[width=0.85\textwidth]{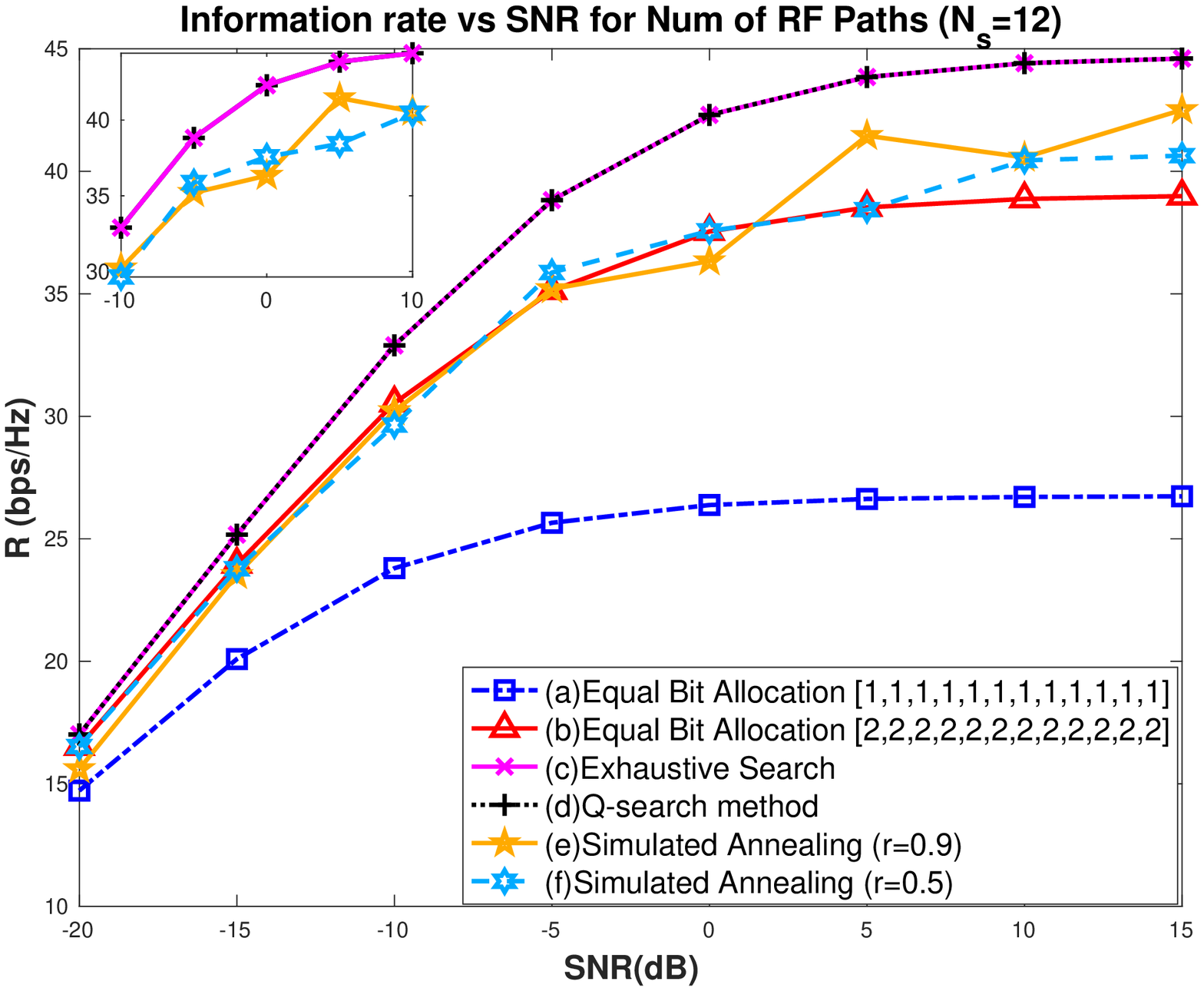}
\captionsetup{justification=centering, font=scriptsize, labelfont={color=Black}}
\caption{\textcolor{Black}{Information rate vs. SNR for $N_s=12$ with 2 dominant scatterers.}}
\label{fig:crlb_R_Nr12_H64by32.eps} 
\end{minipage}
\begin{minipage}[b]{0.45\linewidth}
\centering
\includegraphics[width=0.8\textwidth]{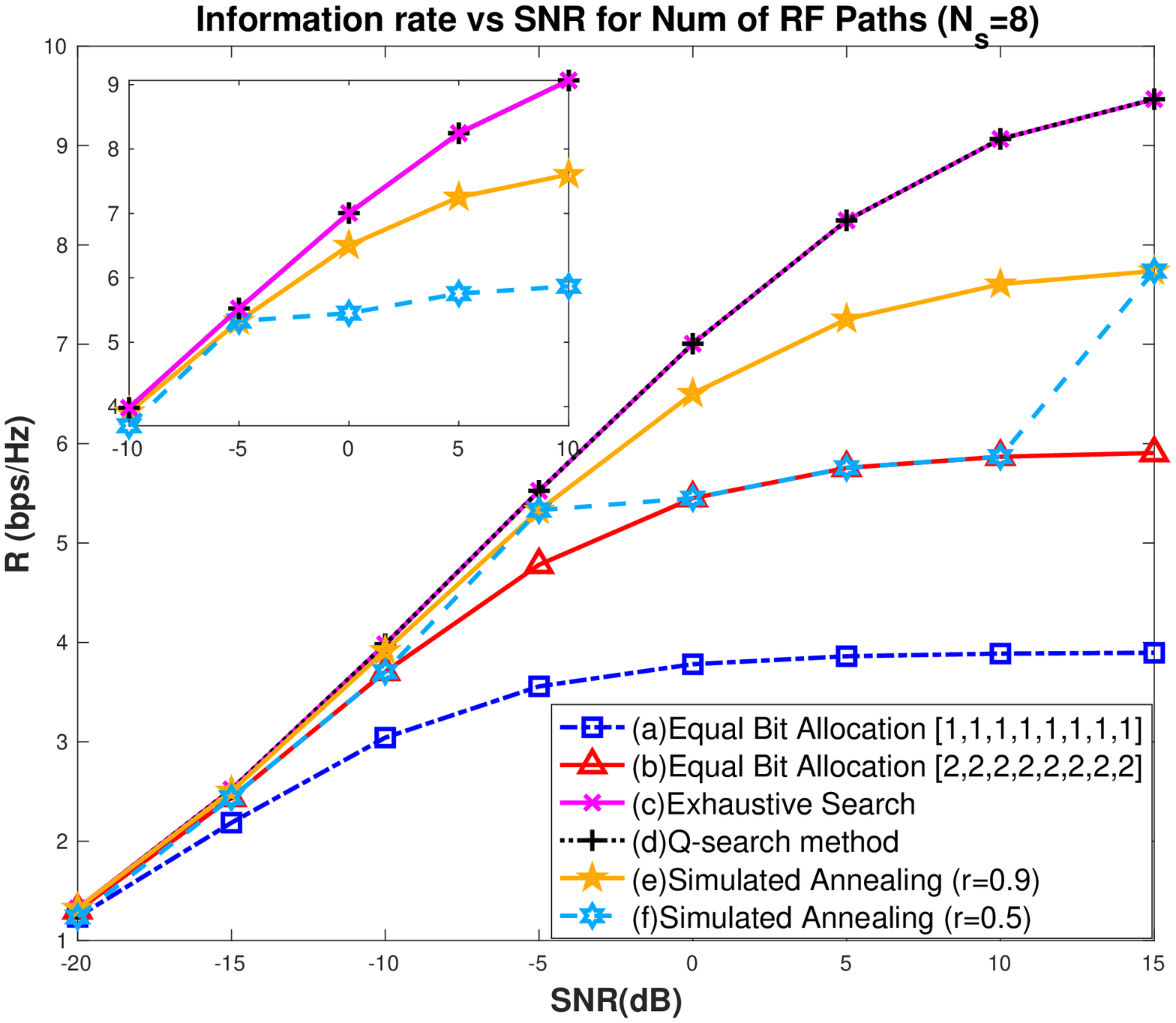}
\captionsetup{justification=centering, font=scriptsize, labelfont={color=Black}}
\caption{\textcolor{Black}{Information rate vs. SNR for $N_s=8$ with a single dominant scatterer.}}
\label{fig:crlb_R_Nr8_H64by32_S2.eps}
\end{minipage}
\quad 
\begin{minipage}[b]{0.45\linewidth}
\centering
\includegraphics[width=0.85\textwidth]{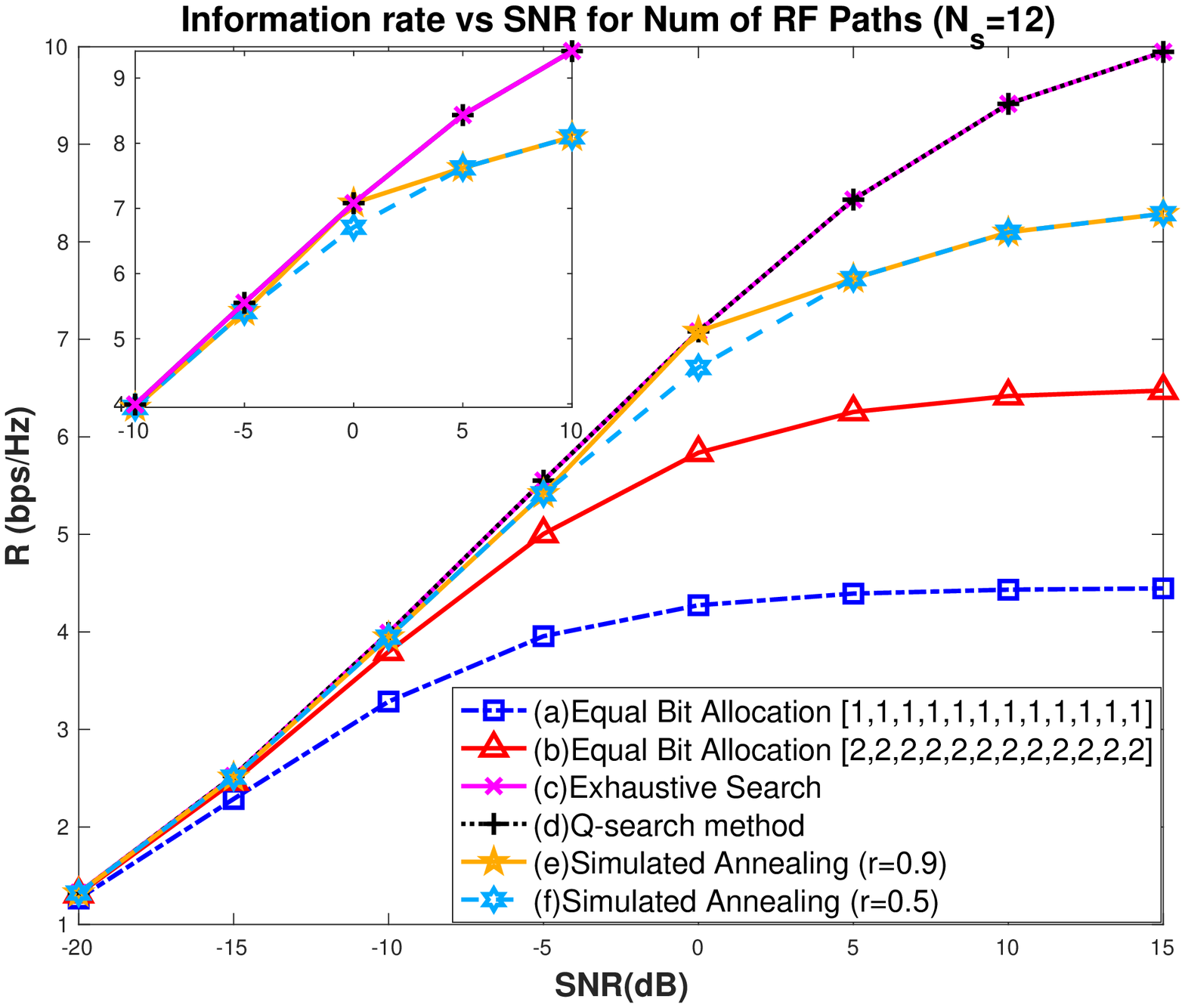}
\captionsetup{justification=centering, font=scriptsize, labelfont={color=Black}}
\caption{\textcolor{Black}{Information rate vs. SNR for $N_s=12$ with a single dominant scatterer.}}
\label{fig:crlb_R_Nr12_H64by32_S2.eps} 
\end{minipage}
\end{figure*}
\begin{table}[!htb]
\begin{center}
\resizebox{\columnwidth}{!}{%
\begin{tabu} to 0.5\textwidth {| l| l| }
 \hline
 \textbf{Parameters}  & \textbf{Value/Type} \\
 \hline
Frequency & 28Ghz \\
\hline
Environment & Line of sight \\
\hline
T-R seperation & 100m\\
\hline
TX/RX array type & ULA\\
\hline 
Num of TX/RX elements $N_t$/$N_r$ & 64/128\\
\hline
TX/RX  antenna spacing & $\lambda/2$\\
\hline
$\eta_{\text{PA}}$ & $40\%$\\
\hline
$P_{\text{CIR}}$ & 10W\\
\hline
$P_{\text{PS}}$ &  50mW\\
\hline
$P_{\text{LNA}}$ &  70mW\\
\hline
$P_{\text{VCO}}$ &  15mW\\
\hline
$c$ &  1432fJ/conversion step \cite{Murmann} \\
\hline
Sampling Frequency &  400Mhz \\
\hline
\end{tabu}}
\caption{$\text{Channel parameters for NYUSIM model \cite{nyusim}}$.} \label{nyusimtab}
\end{center}
\end{table}
\begin{table*}[t]\label{speed_tab}
\centering
\begin{center}
\begin{tabu} to 1.0\textwidth { | p{0.6cm} | X[c] | X[c] | X[c] | X[c] |}
\hline
\multirow{1}{1.5cm}{} &  \multicolumn{4}{c|}{\small Number of complex multiplications}\\
\multirow{2}{1.5cm}{ \small $N_s$} & \multicolumn{4}{c|}{}\\
\cline{2-5}
& \centering \small Exhaustive Search $O(N_b^{N_s})$ \newline \textcolor{red}{High}& \centering \small Q-search method $O(N_s^2)$ \newline \textcolor{Green}{Low} & \centering \small Sim. Annealing (r=0.9) $O(N_s^2)$ \textcolor{Green}{Low} & \small Sim. Annealing (r=0.5)  $O(N_s^2)$ \textcolor{Green}{Low} \\
\hline
\centering \multirow{2}{*}{8}  & \small 1,502,400 & \small 288$^{\S}$ & \small 288$^{\S}$ & \small 288$^{\S}$ \\
& \small 192$^{\S}$ & & &\\
\hline
\centering \multirow{2}{*}{12} & \small 223,865,040 & \small 576$^{\S}$ & \small 576$^{\S}$ &\small 576$^{\S}$ \\
& \small 432$^{\S}$ & & &\\
\hhline{=====}
\multirow{1}{1.5cm}{} &  \multicolumn{4}{c|}{\small Number of complex additions}\\
\multirow{2}{1.5cm}{ \small $N_s$} & \multicolumn{4}{c|}{}\\
\cline{2-5}
& \centering \small Exhaustive Search $O(N_b^{N_s})$ \newline \textcolor{red}{High} & \centering \small Q-search method $O(N_b^{N_s})$ \newline \textcolor{SkyBlue}{Medium} & \centering \small Sim. Annealing (r=0.9) $O(N_s^3)$ \textcolor{SkyBlue}{Medium} & \small Sim. Annealing (r=0.5) $O(N_s^2)$ \textcolor{Green}{Low} \\
\hline
\centering 8  & \small 1,218,822 & 30,272$^{\dagger}$ & \small  2,916$^{\dagger}$ & \small  396$^{\dagger}$ \\
\hline
\centering 12 & \small 193,616,609 & 3,198,552$^{\dagger}$ & \small  6,516$^{\dagger}$ & \small  780$^{\dagger}$ \\
\hline
\end{tabu}
\footnotesize{$^{\S}$ Real multiplications.$^{\dagger}$ Real additions}\\
\caption{Computational complexity in terms of total number of multiplications and additions.} \label{tab:CRLBTab1}
\end{center}
\end{table*}
\section{Computational Complexity Analysis}\label{speed}
In this section, we evaluate the computational complexity in terms of the number of multiplications and additions for the following Algorithms \textit{(i)} ES BA \textit{(ii)} proposed Q-search method \textit{(iii)} proposed SA Algorithm with two cooling factors.\\
\textbf{\textit{(i)} ES Bit-Allocation:}
It can be seen that ES BA requires $\gamma \big(N_s^2 + 2N_s\big)$ complex multiplications, $3N_s^2$ real multiplications, and $\gamma \big(N_s(N_s-1) + N_s\big)$ complex additions. Here $\gamma$ is the number of EE ($\eta_{EE}$) evaluations and is approximately the cardinality of $B_{\text{set}}$, which is $N_b^{N_s}$. Thus ES BA has a multiplicative and additive complexity of $O(N_b^{N_s})$ and thus is NP-Hard.\\
\textbf{\textit{(ii)} Q-search Method:}
The term $q(b_i)$ in \eqref{optee_sol} is precomputed for given $N_b$ and $N_s$. This consists of a table of $N_b \times N_s$ real values $\text{Q}(N_b,N_s)$. This requires the computation of $l_i = {\rm diag}[ {\bold{W}_D^H}{\bold{\Sigma}}^2{\bold{W}_D}+{\bold{I}_{N_s}}]$ that require  in $3N_s^2$ real multiplications and $2N_s^2 + N_s(N_s-1)$ real additions. To compute ${K_{f}}(\bold{b}_j)$ and  $\text{Ptot}()$ as described in Algorithm \ref{AlgoQS} for all BA's in $B_{\text{set}}$ we require 2$\mu\big(N_s+1 \big)$ real additions. Thus, a total of $3N_s^2 + 3N_sN_b$ real multiplications and $3N_s^2 + N_sN_b + \mu \big(N_s-1 \big)$ real additions are required. Here $\mu$ is the number of evaluations of ${K_{f}}(\bold{b}_j)$, which is approximately the cardinality of $B_{\text{set}}$, which is $N_b^{N_s}$.\\
\indent The table consisting of the term $-\log_2(2cf_s) - \log_2( \frac{P_T + P_R}{2cf_s} + \sum_{i=1}^{N_s} 2^{b_i})$ is precomputed and stored as $\text{Ptot}()$ for all BA's in $B_{\text{set}}$. This only requires additions and no multiplications. The term $\frac{P_T + P_R}{2cf_s}$ is independent of BA. The term $\sum_{i=1}^{N_s} 2^{b_i}$ is effectively computed as $\sum_{i=1}^{N_s}\textcolor{black}{\Call{ShiftLeft}{1,{b_i}}}$. The $\log_2()$ can be performed using shift operation and a lookup table \cite{NumC}. The ratio $\frac{R(\bold{b})}{p(\bold{b})}$ is computed without multiplication as illustrated on the line-9 of Algorithm \ref{AlgoQS}.
Thus Q-search method suffers from considerable additive complexity of $O(N_b^{N_s})$. However, it has an order of magnitude reduction in multiplicative complexity, which is $O(N_s^2)$ compared to ES BA. Besides, the Q-search method requires only real multiplications.\\
\textbf{\textit{(iii)} SA Algorithm:}
The terms $\text{Q}(N_b,N_s)$ and $\text{Ptot}()$ is precomputed and stored similar to the Q-search method. Thus resulting in $3N_s^2$ real multiplications and $2N_s^2 + N_s(N_s-1)$ real additions. However, in SA the ${K_{f}}(\bold{b}_j)$ is not evaluated for all $b_j^{'s} \in B_{\text{set}}$ as in Q-search method. The number of evaluations ($\mu$) of ${K_{f}}(\bold{b}_j)$ depends on the initial temperature $T_0$ and the cooling factor $r$. From Algorithm \ref{AlgoCRLB}, it can be seen that $\mu = \ceil*{\frac{\log{\frac{1}{T}}}{\log{r}}}$ and this results in $m\Big\{\ceil*{\frac{\log{\frac{1}{T}}}{\log{r}}}+1\Big\}(2N_s+5)$ real additions. Here $m$ is the number of search at a given temperature $t$. Hence, the additive complexity of SA can be tuned to $O(N_s^D)$ using the parameters $T$ and $r$. The complexity degree of $N_s$ is $D$ and can be derived using the relationship $T \triangleq r^{-N_s^{D-1}}$. In our simulations, we fix $T$ and set $r=0.9$ and $r=0.5$ that correspond to additive complexity of $O(N_s^3)$ and $O(N_s^2)$, respectively. The generation of random numbers is carefully designed and has $O(1)$ complexity. The computation of the acceptance probability $P_a$, which is a sigmoid function is a lookup table with $O(1)$. In conclusion, the SA Algorithm has a real-multiplication complexity of $O(N_b^2)$ and an additive complexity that depends on the initial temperature $T$ and cooling factor $r$.

\section{Conclusion}\label{conc}
For a given power budget, an EE-optimal receiver with reduced computational complexity is crucial to meet the targets set by the 5G standards in terms of the network energy efficiency and spectral efficiency. In this paper, we propose an EE-optimal BA algorithm whose solution is precisely the same as the exhaustive search, with an order of magnitude improvement in multiplicative complexity. Also, we propose a heuristic algorithm using simulated annealing, whose parameters can be tuned to trade off EE optimality with computational complexity. \textcolor{black}{Both algorithms are based on our optimal EE conditions expressed as a function of BA under a power constraint.} We analyze the computational complexities of the proposed methods against ES. The computational complexity of SA is significantly lower than the Q-search method. However, this comes at the cost of no optimality guarantees.
\appendix
\section{Appendix}\label{FirstAppendix}
\renewcommand{\thesubsection}{\Alph{subsection}}
\numberwithin{equation}{section}
\newtheorem{theorem}{Theorem}
\begin{theorem}\label{Thm3}
If $\bold{n}_1 = {\bold{W}_D^H}{\bold{W}_{\alpha}}{\bold{W}_A^H}{\bold{n}} + {\bold{W}_D^H}{\bold{n}_q}$, where $\bold{n}$ is $\bold{n} \sim \mathcal{CN}(\bold{0},{\sigma_n^2\bold{I}_{N_s}})$ and ${\bold{n}_q} \sim \mathcal{N}(\bold{0},{\bold{D}_q^2})$ with ${\bold{D}_q^2} = {\bold{W}_{\alpha}}{\bold{W}_{1-\alpha}}{\diag}[ {\bold{W}_A^H}{\bold{H}}({\bold{W}_A^H}{\bold{H}})^H+{\bold{I}_{N_s}}]$, then it can be shown that $\bold{n}_1$ is a circularly symmetric complex Gaussian (CSCG) vector. That is, $\bold{n}_1 \sim \mathcal{CN}(\bold{0},\bold{\Phi})$.
\end{theorem}
\begin{proof}
The condition for the random vector $\bold{n}_1$ to be CSCG is \cite{RobertGallager}
\begin{equation}\label{proof_eq1}
E[\bold{n}_1] = E[\bold{n}_1\bold{n}_1^T] = \bold{0}.
\end{equation}
Here, $E[\bold{n}_1\bold{n}_1^T]$ is the pseudo-covariance. We first prove that $\bold{n}_q$ is CSCG distributed as $\bold{n}_q \sim \mathcal{N}(\bold{0},{\bold{D}_q^2})$. Given ${\bold{D}_q^2} = E[\bold{n}_q{\bold{n}_q^H}] = {\bold{W}_{\alpha}}{\bold{W}_{1-\alpha}}{\diag}[ {\bold{W}_A^H}{\bold{H}}({\bold{W}_A^H}{\bold{H}})^H+{\bold{I}_{N_s}}]$; with $\bold{W}_{\alpha}$, $\bold{W}_{1-\alpha}$ and ${\diag}[ {\bold{W}_A^H}{\bold{H}}({\bold{W}_A^H}{\bold{H}})^H+{\bold{I}_{N_s}}]$ being positive real diagonal matrices, effectively results in the covariance matrix ${\bold{D}_q^2}$ being positive real diagonal.\\
A necessary and sufficient condition for a random vector $\bold{n}_q$ to be a CSCG random vector is that it has the form $\bold{n}_q = \bold{A}\bold{w}$ where $\bold{w}$ is iid complex Gaussian, that is $\bold{w} \sim \mathcal{CN}(\bold{0},\bold{I}_{N_s})$ and $\bold{A}$ is an arbitrary complex matrix \cite{Vishwa, RobertGallager}.
Since ${\bold{D}_q^2}$ is a positive real diagonal matrix, we can express
\begin{equation}\label{proof_eq2}
\bold{n}_q = \bold{D}_q\bold{w},
\end{equation}
where $\bold{w} \sim \mathcal{CN}(\bold{0},\bold{I}_{N_s})$. This leads to $E[\bold{n}_q] = \bold{D}_qE[\bold{w}] = \bold{0}$ and $E[\bold{n}_q\bold{n}_q^T] = \bold{D}_qE[\bold{w}\bold{w}^T]\bold{D}_q = \bold{0}$. Hence $\bold{n}_q$ is circularly symmetric jointly Gaussian random vector. $\bold{n}_q \sim \mathcal{CN}(\bold{0},{\bold{D}_q^2})$.\\
\noindent
Using \eqref{proof_eq2}, we can express $\bold{n}_1$ as
\begin{equation}\label{proof_eq4}
\begin{split}
\bold{n}_1 = {\bold{W}_D^H}{\bold{W}_{\alpha}}{\bold{W}_A^H}{\bold{n}} + {\bold{W}_D^H}\bold{D}_q{\bold{w}}
\end{split}
\end{equation}
Since we have $\bold{n}$ and $\bold{w}$ as i.i.d complex Gaussian vectors, we can write
\begin{equation}\label{proof_eq5}
\begin{split}
&E[\bold{n}\bold{n}^T] = E[\bold{w}\bold{n}^T] = E[\bold{n}\bold{w}^H] = E[\bold{w}\bold{n}^H] = \bold{0},\\
&E[\bold{n}\bold{n}^H] = \sigma_n^2\bold{I}_{N_s}, \text{   }E[\bold{w}\bold{w}^H] = \bold{I}_{N_s}.\\
\end{split}
\end{equation}
Thus, we arrive at
\begin{equation}\label{proof_eq6}
\begin{split}
E[\bold{n}_1] &= {\bold{W}_D^H}{\bold{W}_{\alpha}}{\bold{W}_A^H}E[{\bold{n}}] + {\bold{W}_D^H}\bold{D}_qE[{\bold{w}}] = 0.\\
E[\bold{n}_1\bold{n}_1^T] &= \bold{G}E[\bold{n}\bold{n}^T]\bold{G}^T \\
&+ \bold{G}E[\bold{n}\bold{w}^T]\bold{D}_q\bold{W}_D + \bold{W}_D^T\bold{D}_qE[\bold{w}\bold{n}^T]\bold{G}^T\\
&+ \bold{W}_D^T\bold{D}_qE[\bold{w}\bold{w}^T]\bold{D}_q\bold{W}_D=\bold{0}.
\end{split}
\end{equation}
Also, 
\begin{equation}\label{proof_eq7}
\begin{split}
E[\bold{n}_1\bold{n}_1^H] = \bold{\Phi} &= \bold{G}E[\bold{n}\bold{n}^H]\bold{G}^H+\bold{G}E[\bold{n}\bold{w}^H]\bold{D}_q\bold{W}_D \\
+ \bold{W}_D^H\bold{D}_qE[\bold{w}\bold{n}^H]\bold{G}^H &+\bold{W}_D^H\bold{D}_qE[\bold{w}\bold{w}^H]\bold{D}_q\bold{W}_D,\\
&= \sigma_n^2\bold{G}\bold{G}^H + \bold{W}_D^H\bold{D}_q^2\bold{W}_D.
\end{split}
\end{equation}
Thus, $\bold{n}_1 \sim \mathcal{CN}(\bold{0},\bold{\Phi})$ is a CSCG vector.
\end{proof}
\vspace{-8mm}
\newtheorem{lem}{Lemma}
\begin{lem}\label{lemm1}
The term  $\log_2  \Big( q(b_i) + 1 \Big)$ for $0 \leq q(b_i) < 1$, can be approximated as $\log_2  \Big( q(b_i) + 1 \Big) \simeq \frac{q(b_i)}{\ln2}$. 
\end{lem}
\begin{proof}
We can write:\\
$\log_2  \Big( \frac{p\sigma_i^2}{\sigma_n^2 + g(b_i)l_i} + 1 \Big) = \frac{1}{\ln 2}{\ln \Big( \frac{p\sigma_i^2}{\sigma_n^2 + g(b_i)l_i} + 1 \Big)}$.\\
\indent
We can approximate $g(b_i)$ as $c2^{-db_i}$, where $d=2.0765, c=2.40667$. For the sake of simplicity, we will replace the variable $\bold{b} \in \mathbb{I}^{N_s \times 1}$ with $\bold{x} \in \mathbb{R}^{N_s \times 1}$.\\
\indent
We will now define $f\big(p(x_i)\big) = \ln \Big( \frac{p\sigma_i^2}{\sigma_n^2 + c2^{dx_i}l_i} + 1 \Big)$, where $p(x_i) = \frac{p\sigma_i^2}{\sigma_n^2 + c2^{dx_i}l_i}$. For a geometric series below, with a common ratio of $-p(x_i)$, where $0 \leq p(x_i) < 1$, we can write
\begin{equation}\label{apx_c1}
1-p(x_i)+p(x_i)^2-p(x_i)^3+.. = \frac{1}{1+p(x_i)}.
\end{equation}
\begin{equation}\label{apx_c2}
\ln(1 + p(x_i)) = \int \frac{1}{1 + p(x_i)} d(p(x_i)),
\end{equation}
substituting for $\frac{1}{1+p(x_i)}$ into the integral in $\ref{apx_c2}$ from $\ref{apx_c1}$, we have
\begin{equation}\label{apx_c3}
\ln(1 + p(x_i)) = p(x_i)-\frac{p(x_i)^2}{2}+\frac{p(x_i)^3}{3}-\frac{p(x_i)^4}{4}+...
\end{equation}
Given that $0 \leq p(x_i) < 1$, the higher powers of $p(x_i)$ are negligible and thus the above series can be approximated  as
\begin{equation}\label{apx_c4}
f\big(p(x_i)\big)  \simeq p(x_i).
\end{equation}
By re-substituting variable $\bold{x} \in \mathbb{R}^{N_s \times 1}$ with $\bold{b} \in \mathbb{I}^{N_s \times 1}$, we can effectively write
\begin{equation}\label{apx_c5}
\log_2  \Big( \frac{p\sigma_i^2}{\sigma_n^2 + g(b_i)l_i} + 1 \Big) \simeq \frac{1}{\ln2}\Big(\frac{p\sigma_i^2}{\sigma_n^2 + g(b_i)l_i}\Big). 
\end{equation}
\end{proof}
\vspace{-0.3in}
\begin{lem}\label{lemm2}
It can be shown that $\log_2  \Big( q(b_i) + 1 \Big) = \Big(1-\frac{1}{q(b_i)}\Big)P + L(p,\sigma_i^2, \sigma_n^2)$ for $\infty > q(b_i) \geq 1$, where the terms $P$ and $L(p,\sigma_i^2, \sigma_n^2)$ are not functions of $b_i$.
\end{lem}
\begin{proof}
Consider the expansion for $f\big(p(x_i)\big)$ for $\infty > p(x_i) \geq 1$. We can approximate $f\big(p(x_i)\big)$ as
\begin{equation}\label{apx0}
f(\big(p(x_i)\big) = \ln \Big( p(x_i) + 1 \Big) \simeq \ln \Big(p(x_i) \Big).
\end{equation}
Rewriting $f(\big(p(x_i)\big)$ as:
\begin{equation}\label{apx1}
\begin{split}
&f(\big(p(x_i)\big) = -\ln \bigg( \frac{1}{p(x_i)} \bigg) \text{ for } 0 < \frac{1}{p(x_i)} \leq 2;\\
&f(\big(p(x_i)\big) = -\ln \Big(g(x_i)\Big) \text{ where } g(x_i) = \frac{1}{p(x_i)};\\
\text{ or }& f(\big(p(x_i)\big) = -h\big(g(x_i)\big) \text{ where } h\big(g(x_i)\big) = \ln\big(g(x_i)\big);
\end{split}
\end{equation}
Using the Taylor series at $g(x_i = x_0) = 1 = \frac{1}{p(x_i = x_0)}$ with the region of convergence $R: \infty > p(x_i) \geq \frac{1}{2}$, we have
\begin{equation}\label{apx2}
\begin{split}
&h\big(g(x_i)\big) = h\big(g(x_0)\big) + h^\prime\big(g(x_0)\big)(g(x_i)-1)\\ 
&+ \frac{1}{2}h^{\prime\prime}\big(g(x_0)\big)(g(x_i)-1)^2 + \frac{1}{6}h^{\prime\prime\prime}\big(g(x_0)\big)(g(x_i)-1)^3 + ..
\end{split}
\end{equation}
Also:
\begin{equation}\label{apx3}
\begin{split}
h\big(g(x_0)\big) &= \ln(1) = 0; \text{  }\\
h^\prime\big(g(x_i)\big) &= \frac{1}{g(x_i)} \implies h^\prime\big(g(x_0)\big) = 1;\\
h^{\prime\prime}\big(g(x_i)\big) &= -\frac{1}{[g(x_i)]^2},  h^{\prime\prime}\big(g(x_0)\big) = -1;\\
h^{\prime\prime\prime}\big(g(x_i)\big) &= \frac{2}{[g(x_i)]^3}, h^{\prime\prime\prime}\big(g(x_0)\big) = 2;\cdots
\end{split}
\end{equation}
substituting $\ref{apx3}$ in $\ref{apx2}$, we have
\begin{equation}\nonumber
\begin{split}
h\big(g(x_i)\big) &= \Big(\frac{1}{p(x_i)}-1\Big) - \frac{1}{2}\Big(\frac{1}{p(x_i)}-1\Big)^2
\end{split}
\end{equation}
\begin{equation}\label{apx4}
\begin{split}
&+\frac{1}{3}\Big(\frac{1}{p(x_i)}-1\Big)^3 - ..\\
f\big(p(x_i)\big) &= \Big(1-\frac{1}{p(x_i)}\Big) - \sum_{n=2}^{\infty} \frac{(-1)^{(n-1)}}{n}\Big(\frac{1}{p(x_i)}-1\Big)^n
\end{split}
\end{equation}
Using binomial expansion for $\Big(\frac{1}{p(x_i)}-1\Big)^n$, we can write
\begin{equation}\label{apx6a}
\Big(\frac{1}{p(x_i)}-1\Big)^n = \sum_{k=0}^n {{n}\choose{k}}\frac{-1^{(n-k)}}{(p(x_i))^k} = K_n(p,\sigma_i^2, \sigma_n^2).
\end{equation}
It is to be noted that for $n \ge 2$ and larger values of $k$, the term $K_n(p,\sigma_i^2, \sigma_n^2)$ becomes less dependent on $x_i$ and is convergent for $p(x_i) \ge 1$. So, we can write $\ref{apx6a}$ as
\begin{equation}\label{apx5}
f\big(p(x_i)\big) = \Big(1-\frac{1}{p(x_i)}\Big) + G(p,\sigma_i^2, \sigma_n^2),
\end{equation}
Where $G(p,\sigma_i^2, \sigma_n^2) = - \sum_{n=2}^{\infty} \frac{(-1)^{(n-1)}K_n(p,\sigma_i^2, \sigma_n^2)}{n}$ and is a converging series. By re-substituting variable $\bold{x} \in \mathbb{R}^{N_s \times 1}$ with $\bold{b} \in \mathbb{I}^{N_s \times 1}$, we can effectively write
\begin{equation}\label{apx6}
\log_2  \bigg( \frac{p\sigma_i^2}{\sigma_n^2 + g(b_i)l_i} + 1 \bigg) = P\bigg(1-\frac{1}{\frac{p\sigma_i^2}{\sigma_n^2 + g(b_i)l_i}}\bigg) + L(p,\sigma_i^2, \sigma_n^2).
\end{equation}
where $P = \frac{1}{\ln2}$ and $L(p,\sigma_i^2, \sigma_n^2) = \frac{G(p,\sigma_i^2, \sigma_n^2)}{\ln2}$.
\end{proof}



%
\bibliographystyle{IEEEtran}
\bibliography{Opt_BA_MaMIMO_ITGCN}
\end{document}